\documentclass[lettersize,journal]{IEEEtran}
\usepackage{amsmath,amssymb,amsfonts}
\usepackage{algorithmic}
\usepackage{arydshln}
\usepackage{algorithm}
\usepackage[caption=false,font=normalsize,labelfont=sf,textfont=sf]{subfig}
\usepackage{color}
\usepackage{balance}
\usepackage{textcomp}
\usepackage{stfloats}
\usepackage{url}
\usepackage{verbatim}
\usepackage{graphicx}
\usepackage{cite}
\usepackage{bm} 
\usepackage{mathtools}
\usepackage{caption}

\newtheorem{theorem}{Theorem}
\newtheorem{lemma}{Lemma}
\newtheorem{definition}{Definition}
\newtheorem{example}{Example}
\newtheorem{corollary}{Corollary}
\newtheorem{remark}{Remark}
\newenvironment{proof}{\quad{\noindent\it Proof:}}{\hfill $\blacksquare$\par}

\begin{document}

\title{Controllability of Multilayer Networked Sampled-data Systems}

\author{Zixuan Yang, Xiaofan Wang, \IEEEmembership{Senior Member, IEEE}, and Lin Wang, \IEEEmembership{Senior Member, IEEE}
\thanks{This work was supported by the National Natural Science Foundation of China (No. 61873167, 61773255), and in part by Shuguang Project of Shanghai Education Commission, Shanghai Talent Development Fund (No. 2021011), Automotive Industry Science
and Technology Development Foundation of Shanghai (No.1904) and the
Strategic Priority Research Program of Chinese Academy of Sciences (No.
XDA27000000). \textit{(Corresponding author: Lin Wang.)}}
\thanks{Zixuan Yang and Lin Wang are with the Department of Automation, Shanghai Jiao Tong University, and Key Laboratory of System Control and Information Processing, Ministry of Education of China, Shanghai 200240, China (e-mail: jennifer\_yang@sjtu.edu.cn; wanglin@sjtu.edu.cn). }
\thanks{Xiaofan Wang is with the Department of Automation, Shanghai Jiao Tong University, Shanghai 200240, China, and also with the Department of Automation, Shanghai University, Shanghai 200444, China (e-mail: xfwang@sjtu.edu.cn). }}



\maketitle

\begin{abstract}
This paper explores the state controllability of multilayer networked sampled-data systems with inter-layer couplings, where zero-order holders (ZOHs) are on the control and transmission channels.
The effects of both single- and multi-rate sampling on controllability of multilayer networked linear time-invariant (LTI) systems are analyzed, with some sufficient and/or necessary controllability conditions derived.
Under specific conditions, the pathological sampling of single node systems could be eliminated by the network structure and inner couplings among different nodes and different layers.
The representative drive-response inter-layer coupling mode is studied, and it reveals that the whole system could be controllable due to the inter-layer couplings even if the response layer is uncontrollable itself.
Moreover, simulated examples show that the modification of sampling rate on local channels could lay a positive or negative effect on the controllability of the whole system.
All the results indicate that the controllability of the multilayer networked sampled-data system is collectively affected by mutually coupled factors.
\end{abstract}


\section{Introduction}
\IEEEPARstart{A}s an important prerequisite of effective system control, controllability has been extensively investigated since the 1960’s, with various rank criteria and graphic properties achieved \cite{xiang2019advances,chen2012optimal,kalman1962canonical,davison1975new,gilbert1963controllability,kobayashi1978controllability,lin1974structural,liu2011controllability,menichetti2014network,liu2012control}.
Recent years have witnessed an unprecedented upsurge of network science and information technology.
As a result, the scale of real-world systems has been expanded, where node states are higher-dimensional and complexly coupled with each other through multiple transmission channels.
For these networked systems, in \cite{zhou2015controllability} controllability conditions are derived based on the transfer function matrix, while an easier-to-verify criterion is developed in \cite{hao2019new} by matrix similarity transformation. 
In \cite{wang2016controllability}, it is claimed that controllability of the networked system is jointly determined by the coupling of network structure and node dynamics.
In \cite{iudice2019node}, a controllability decomposition approach is provided to analyze each node system when the network is not completely controllable.
Research also shows that the controllability of a special type of networked systems, multi-agent systems (MASs), can be decoupled into two independent parts related to single nodes and network topology, respectively \cite{ji2014protocols,ni2013consensus,zhao2020data}.\par

In networked systems, the interactions among different layers increase network complexity and bring new challenges to controllability research.
A target path-cover algorithm based on maximum flow is put forward in \cite{song2019target} to guarantee target controllability of two-layer multiplex networks with minimum control sources.
In \cite{posfai2016controllability}, the underlying mechanisms connecting controllability and time-scale difference between two network layers is identified.
In \cite{su2018controllability}, the two-time-scale system is detached into fast and slow subsystems by the iterative method and approximate approach.
A compositional framework is proposed in \cite{chapman2014controllability} to find out the controllability of composite network-of-network from the corresponding factor networks.
And in \cite{hao2019controllability}, a modified controllability condition is developed, where the diagonalizability requirement for the topology matrix of composite networks is removed.
The collective effects of intra-layer couplings and inter-layer dynamics on controllability of deep-coupling networks is explored in\cite{wu2020controllability}, and furthermore, different coupling modes are considered in \cite{jiang2021controllability} to represent multiple connection structures of the network.\par
\IEEEpubidadjcol

Nowadays, with the development of digital platforms, information is mostly transmitted in the form of sampled data. 
Considering the bandwidth limitation, signal instability, delay and other factors that exist in practice, the controllability of the sampled-data systems is also worth studying.
It is clarified in \cite{chen2012optimal} that the controllability of a single continuous systems can be damaged after pathological periodic sampling.  
The effects of sampling on controllability indices are analyzed in \cite{hagiwara1988controllability}.
In \cite{kreisselmeier1999sampling}, a step of non-equidistant sampling is added to maintain the controllability of systems after sampling, which is further applied to time-varying systems \cite{guo2005systems}.
For multi-rate sampling, the case of different sampling periods on different channels was studied in \cite{pasand2018controllability} with a sufficient controllability condition given.
However, there are few researches on controllability of networked sampled-data system.
Although the sampling controllability of MASs has received attention \cite{ji2014controllability,lu2020sampled}, in reality, more networked systems cannot be decoupled into two independent parts like MASs.
In \cite{my2021}, we have studied the controllability of single-layer networked sampled-data systems and have got some preliminary results.\par

In view of the multilayer structure of real-world networks and system design requirements, this paper studies the controllability of multilayer networked sampled-data systems.
The systems are synthesized by directed, weighted multilayer network topology and multi-dimensional node dynamics. 
On each control channel and inter/intra-layer transmission channel, the information is sampled by a zero-order holder (ZOH).
The controllability verification requires even more intensive calculation due to the more complex network structure, the larger system scale, and the more diverse sampling patterns.
However, the method in this paper is computionally advantageous since it states conditions with respect to decomposed, single-layer, and single-rate systems.
Specifically:
(1) The representation of the multilayer networked sampled-data systems is proposed, with single- and multi-rate sampling patterns considered, respectively. 
(2) Sufficient and/or necessary controllabiltiy conditions are developed, combined with factors of the network topology, external inputs, inner couplings, node dynamics and sampling rates.
(3) The controllability of the response layer is not necessary to the controllability of the whole system due to the inter-layer interactions.
(4) The loss of controllability caused by pathological sampling of single node systems can be eliminated by the multilayer network structure and inter-layer couplings.
(5) The modification of the local sampling rate could damage or enhance the controllability of the whole system.
\par

The rest of this paper is organized as follows.
The notations and model formulation are introduced in Section \ref{sec:pre}.
In Section \ref{sec:gen}, a controllability condition for general multilayer networked sampled-data systems is developed.
Two-layer networked sampled-data systems with drive-response mode are studied in Section \ref{sec:2_dri}, while systems with deeper layers are considered in Section \ref{sec:deep}.
Section \ref{sec:multi} preliminarily inspects the controllability of multilayer networked multi-rate sampled-data systems.
Some useful simulated examples are provided in Section \ref{sec:exa}.
Finally, Section \ref{sec:con} summarizes this paper.

\section{Notations and Model Formulation}\label{sec:pre}
\subsection{Notations}
Denote $\mathbb{R}$, $\mathbb{C}$ and $\mathbb{N}$ fields of real, complex and natural numbers, respectively.
Let $I_n$ denote the identity matrix of size $n\times n$, and by $e_i$ the $i$th unit row vector whose entries are all zero except that the $i$th element is $1$. 
Denote by $diag\{a_1,a_2,...,a_n\}$ the $n\times n$ matrix with diagonal elements $a_1,a_2,...,a_n$, and by $diagblock\{A_1,A_2,...,A_n\}$ the matrix with diagonal block matrices $A_1,A_2,...,A_n$.
The set of all eigenvalues of matrix $A\in\mathbb{R}^{n\times{n}}$ is denoted by $\sigma(A)=\{\lambda_1,...,\lambda_r\}$, $1\leq{r}\leq{n}$, where $r$ is the sum of the geometric multiplicity of all eigenvalues of $A$, and $M(\lambda_i|A)$ denotes the eigenspace of $A$ with respect to $\lambda_i$.
The complex linear span of row vectors $v_1,v_2,...,v_n$ is denoted by $span\{v_1,v_2,...,v_n\}=\{\Sigma_{i=1}^n c_iv_i|c_i\in\mathbb{C}\}$, which is the set of their all complex linear combinations.
Let $A\otimes{B}$ denote the Kronecker product of matrices $A$ and $B$, and $V_1\oplus{V_2}$ the direct sum of space $V_1$ and $V_2$.
Denote $\mathbf{0}$ and $O$ the zero vector and zero matrix, respectively.
Assume that the dimensions of matrices are compatible for algebraic operations if they are not specified.\par

\subsection{Model Formulation}
Consider a general directed and weighted network consisting of $M$ layers with $N$ identical node systems in each layer.
The dynamics of the $i$th node in the $K$th layer are described as:
\begin{equation}
\label{deep_coupling}
\begin{aligned}
    \dot{x}_i^K(t)&=A^Kx_i^K(t)+\sum_{j=1}^N{w_{ij}^KH^KC^Kx_j^K(t)}\\
    &+\sum_{L=1,L\neq{K}}^M\sum_{j=1}^N{d_{ij}^{K,L}P^{K,L}C^Lx_j^L(t)}+\delta_i^KB^Ku_i^K(t)
\end{aligned}
\end{equation}
where $i=1,...,N$, and $K=1,...,M$.
Suppose that $N, M\geq{2}$.
$x_i^K(t)\in\mathbb{R}^n$ and $u_i^K(t)\in\mathbb{R}^p$ denote the state vector and input vector of node $i$ in layer $K$, respectively.
$A^K\in\mathbb{R}^{n\times{n}},B^K\in\mathbb{R}^{n\times{p}},C^K\in\mathbb{R}^{m\times{n}}$ denote the state matrix, input matrix and output matrix of nodes in layer $K$, respectively.
$H^K\in\mathbb{R}^{n\times{m}}$ denotes the inner-couplings among nodes in layer $K$, while $P^{K,L}\in\mathbb{R}^{n\times{m}}$ denotes the inner-couplings between nodes in layer $K$ and nodes in layer $L$.\par

Let $W^K=[w^K_{ij}]\in\mathbb{R}^{N\times{N}}$ denote the intra-layer network topology of layer $K$, where $w_{ii}=0$, and $w^K_{ij}\neq{0}$ ($i\neq{j}$) if there is a link from node $j$ to node $i$ in layer $K$, otherwise $w^K_{ij}=0$.
$D^{K,L}=[d^{K,L}_{ij}]\in\mathbb{R}^{N\times{N}}$ describes the inter-layer coupling topology, where $d^{K,L}_{ij} \neq{0}$ if there is a link from node $j$ in layer $L$ to node $i$ in layer $K$, otherwise $d^{K,L}_{ij}=0$ .
Define $\Delta^K=diag\{\delta^K_1,... ,\delta^K_N\}$, where $\delta^K_i=1$ if node $i$ in layer $K$ is under control; otherwise, $\delta^K_i=0$.
To obtain a compact form, define $\mathbf{X}=((X^1)^\top,...,(X^M)^\top)^\top$ and $\mathbf{U}=((U^1)^\top,...,(U^M)^\top)^\top$ as the total state and control input of the whole network, respectively, where $X^K=((x^K_1)^\top,...,(x^K_N)^\top)^\top$ and $U^K=((u^K_1)^\top,...,(u^K_N)^\top)^\top$ are the state and input of layer $K$, respectively.
Then the multi-layer networked continuous linear time-invariant (CLTI) system can be written as:
\begin{equation}
\label{multi_layer}
    \dot{\mathbf{X}}(t)=\bm{\Phi}\mathbf{X}(t)+\bm{\Psi}\mathbf{U}(t),
\end{equation}
where 
$$\bm{\Phi}=\begin{bmatrix}
\Phi^{1,1} & \cdots & \Phi^{1,M}\\
 \vdots & \ddots &\vdots \\
 \Phi^{M,1} & \cdots & \Phi^{M,M}
\end{bmatrix},
\bm{\Psi}=\begin{bmatrix}
\Psi^{1,1} &  & \\
  & \ddots & \\
  &  & \Psi^{M,M}
\end{bmatrix},$$

\begin{equation}
\label{multi_layer_detail}
\begin{aligned}
&\Phi^{K,K}=I_N\otimes{A^K}+W^K\otimes{H^K}C^K, K=1,2,...,M,\\
&\Phi^{K,L}=D^{K,L}\otimes{P^{K,L}}C^L, K,L=1,2,...,M,K\neq{L},\\
&\Psi^{K,K}=\Delta^K\otimes{B^K}, K=1,2,...,M.
\end{aligned}
\end{equation}

\begin{figure}[!t]
\centering
\subfloat[]{\includegraphics[width=1.45in]{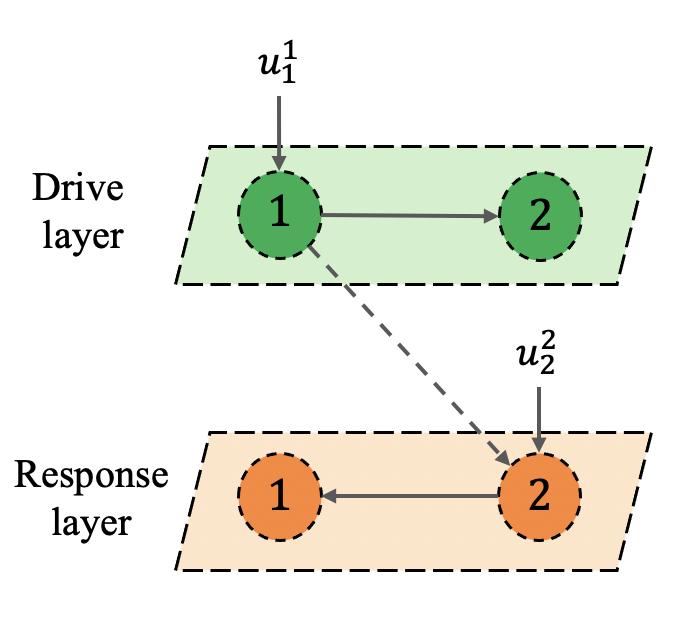}%
\label{fig:CLTI}}
\hfil
\subfloat[]{\includegraphics[width=1.7in]{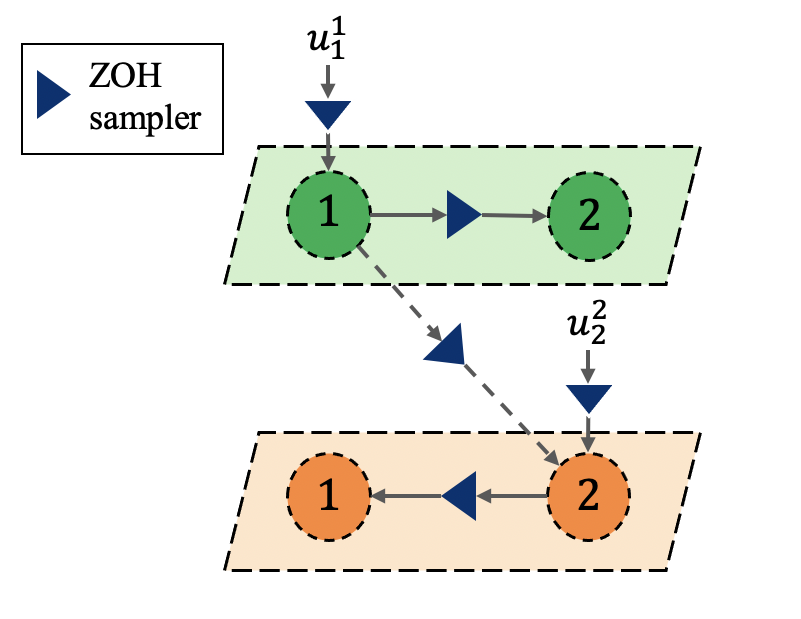}%
\label{fig:sam}}
\caption{An example of a multilayer networked CLTI system and its corresponding sampled-data version. (a) The networked CLTI system. (b) The networked sampled-data system.}
\label{fig:1}
\end{figure}

By performing the ZOHs on control channels and transmission channels simultaneously, the sampled-data version of the multilayer networked system can be obtained, as is shown in Fig.\ref{fig:1}.
Denote the sampling period by $h$, and the system can be written as:
\begin{equation}
\label{sampled-multi-layer}
    \mathbf{X}((k+1)h)=\bm{\Phi}_s\mathbf{X}(kh)+\bm{\Psi}_s\mathbf{U}(kh),
\end{equation}
where $k\in\mathbb{N}$, and
$$\bm{\Phi}_s=\begin{bmatrix}
\Phi_s^{1,1} & \cdots & \Phi_s^{1,M}\\
 \vdots & \ddots &\vdots \\
 \Phi_s^{M,1} & \cdots & \Phi_s^{M,M}
\end{bmatrix},
\bm{\Psi}_s=\begin{bmatrix}
\Psi_s^{1,1} &  & \\
  & \ddots & \\
  &  & \Psi_s^{M,M}
\end{bmatrix},$$
\begin{equation}
\label{sampled-multi-layer-detail}
\begin{aligned}
&\Phi_s^{K,K}=I_N\otimes{e^{A^Kh}}+W^K\otimes\mathcal{H}^K(h), K=1,2,...,M,\\
&\Phi_s^{K,L}=D^{K,L}\otimes\mathcal{P}^{K,L}(h), K,L=1,2,...,M,K\neq{L},\\
&\Psi_s^{K,K}=\Delta^K\otimes\mathcal{B}^K(h), K=1,2,...,M.
\end{aligned}
\end{equation}
Note that $\mathcal{P}^{K,L}(h)=\int_0^h e^{A^K\tau}d\tau{P^{K,L}}C^L$, $\mathcal{H}^K(h)=\int_0^h e^{A^K\tau}d\tau{H^K}C^K$, $\mathcal{B}^K(h)=\int_0^h e^{A^K\tau}d\tau{B^K}$
are denoted for simplicity.

\section{General Multi-layer Networked Sampled-data Systems}\label{sec:gen}
As a premise for subsequent analysis, a more general result is given at first.
Similar to the necessary and sufficient controllability condition for the CLTI system (\ref{multi_layer}-\ref{multi_layer_detail}) in \cite{wu2020controllability}, here a sufficient condition is derived for the multilayer networked sampled-data system (\ref{sampled-multi-layer}-\ref{sampled-multi-layer-detail}).
The lack of necessity is because the reachable subspace and controllable subspace of the discrete-time system are not equivalent.

\begin{theorem}
The multi-layer networked sampled-data system (\ref{sampled-multi-layer}-\ref{sampled-multi-layer-detail}) is controllable if, $\forall{s}\in\mathbb{C}$, the following matrix equations 
\begin{equation}\label{gen_1}
\begin{aligned}
&F^K(sI_n-e^{A^Kh})-(W^K)^\top{F^K}\mathcal{H}^K(h)\\=&\sum_{L=1,L\neq{K}}^M (D^{L,K})^\top{F}^L\mathcal{P}^{L,K}(h),
\end{aligned}    
\end{equation}
\begin{equation}\label{gen_2}
\Delta^K{F}^K\mathcal{B}^K(h)=O
\end{equation}
have a unique solution of $F^K=0$ for every $K=1,...,M$, where $F^K\in\mathbb{C}^{N\times{n}}$.
\end{theorem}

\begin{proof}
According to the PBH rank condition, system (\ref{sampled-multi-layer},\ref{sampled-multi-layer-detail}) is controllable if $[sI_{MNn}-\bm{\Phi}_s,\bm{\Psi}_s]$ is of full rank for $\forall{s}\in\mathbb{C}$, that is, 
\begin{equation}\label{proof1}
    \begin{aligned}
    &[f^1,...,f^M](sI_{MNn}-\bm{\Phi}_s)=\mathbf{0}\\
    &[f^1,...,f^M]\bm{\Psi}_s=\mathbf{0}
    \end{aligned}
\end{equation}
hold only if $f^K=\mathbf{0}$ for every $K=1,...,M$, where $f^K=[f^K_1,...,f^K_N],f^K_i\in\mathbb{C}^{1\times{n}}$ for $i=1,...,N$.
Let $F^K=[(f_1^K)^\top,...,(f_N^K)^\top]^\top$, then it is easy to find that equations (\ref{proof1}) are equivalent to equations (\ref{gen_1}-\ref{gen_2}).
Thus, system (\ref{sampled-multi-layer}-\ref{sampled-multi-layer-detail}) is controllable if $\forall{s}\in\mathbb{C}$, equations (\ref{gen_1}-\ref{gen_2}) have a unique solution $F^K=O$ for every $K=1,...,M$.
\end{proof}

\section{Drive-response Mode between Two Layers}\label{sec:2_dri}
Consider a two-layer networked sampled-data system, where a ZOH is added on each control and transmission channel.
The inter-layer couplings between two layers are of drive-response mode, which means inter-layer links are only from the drive layer to the response layer.\par

\subsection{General Two-layer Drive-response Mode}
To begin with, consider the general two-layer drive-response mode, where the node systems at different layers are heterogeneous.
In this case, $\bm{\Phi}_s$ and $\bm{\Psi}_s$ in (\ref{sampled-multi-layer}) are
\begin{equation}\label{two_sam}
\bm{\Phi}_s=\begin{bmatrix}
\Phi_s^{1,1} & O\\
\Phi_s^{2,1} & \Phi_s^{2,2}
\end{bmatrix},
\bm{\Psi}_s=\begin{bmatrix}
\Psi_s^{1,1} & O\\
O & \Psi_s^{2,2}
\end{bmatrix},\\    
\end{equation}
\begin{equation}\label{two-layer-driver_s}
\begin{aligned}
&\Phi_s^{1,1}=I_N\otimes{e^{A^1h}}+W^1\otimes\mathcal{H}^1(h),\\
&\Phi_s^{2,1}=D^{2,1}\otimes\mathcal{P}^{2,1}(h),\\ &\Phi_s^{2,2}=I_N\otimes{e^{A^2h}}+W^2\otimes\mathcal{H}^2(h),\\
&\Psi_s^{1,1}=\Delta^1\otimes\mathcal{B}^1(h),
\Psi_s^{2,2}=\Delta^2\otimes\mathcal{B}^2(h).
\end{aligned}
\end{equation}

The eigenvalues and the eigenspace of each layer are given by Lemma \ref{two-driver-value} and Lemma \ref{two-driver-vector}, respectively.

\begin{lemma}\label{two-driver-value}
Let $\sigma(W^K)=\{\lambda_1^K,...,\lambda_{r^K}^K\}$, and $\sigma(E^K_i)=\{\theta_{i,1}^K,...,\theta_{i,p(i)}^K\}$, where $E_i^K=e^{A^Kh}+\lambda_i^K\mathcal{H}^K(h)$, $i=1,...,r^K$.
Then $\sigma(\Phi^{K,K}_s)=\{\theta_{1,1}^K,...,\theta_{1,p(1)}^K,...,\theta_{r^K,1}^K,...,$ $\theta_{r^K,p(r^K)}^K\}$, $K=1,2$.
\end{lemma}

\begin{lemma}\label{two-driver-vector}
Assume that the left Jordan chain of $W^K$ about $\lambda_i^K$ is $v_i^K(1),...,v_i^K(\alpha_i^K)$, and the generalized left Jordan chain of $E_i^K$ about $\mathcal{H}^K(h)$ related to $\theta_{i,j}^K$ is $\xi_{i,j}^K(1),...,\xi_{i,j}^K(\gamma_{i,j}^K)$, where $i\in\{1,2,...,r^K\}$ and $j\in\{1,2,...,p^K(i)\}$.
Then the eigenspace of layer $K$ with respect to $\theta_{i,j}^K$ is $M(\theta_{i,j}^K|\Phi^{K,K}_s)=V(\theta_{i,j}^K)$, where $V(\theta_{i,j}^K)=span\{\eta_{i,j}^K(1),...,\eta_{i,j}^K(\beta_{i,j}^K)\},$ with $\eta_{i,j}^K(1)=v_i^K(1)\otimes\xi_{i,j}^K(1),\eta_{i,j}^K(2)=v_i^K(1)\otimes\xi_{i,j}^K(2)+v_i^K(2)\otimes\xi_{i,j}^K(1),..., \eta_{i,j}^K(\beta_{i,j}^K)=v_i^K(\beta_{i,j}^K)\otimes\xi_{i,j}^K(1)+...+v_i^K(1)\otimes\xi_{i,j}^K(\beta_{i,j}^K),\beta_{i,j}^K=min\{\alpha_i^K,\gamma_{i,j}^K\}$, $i=1,2,...,r^K$ and $j=1,2,...,p^K(i)$.
Specially, if $\theta_{i_1,j_1}^K=...=\theta_{i_q,j_q}^K$, $q>1$, the eigenspace of $\Phi^{K,K}_s$ with respect to $\theta_{i,j}^K$ should be the direct sum of all the eigenspace about $\theta_{i,j}^K$, i.e., $M(\theta_{i,j}^K|\Phi^{K,K}_s)=\oplus_{l=1}^q{V(\theta_{i_l,j_l}^K)}$.
\end{lemma}

\begin{remark}
The notion of generalized Jordan chain in Lemma \ref{two-driver-vector} is learnt from \cite{hao2019new}, as well as the method of decomposing eigenvalues and eigenspace of $\Phi^{K,K}_s$.
It is obvious that $\sigma(\bm{\Phi}_s)=\sigma(\Phi^{1,1}_s)\cup\sigma(\Phi^{2,2}_s)$.
Thus, a sufficient controllability condition can be obtained for system (\ref{sampled-multi-layer},\ref{two_sam}-\ref{two-layer-driver_s}) based on Lemma \ref{two-driver-value} and Lemma \ref{two-driver-vector}.
Theorem \ref{con:two-driver} reveals that the controllability of multilayer networked sampled-data systems is a collective effect of mutually coupled factors such as network topology, intra- and inter-layer couplings, the sampling period, external control inputs and node dynamics.
\end{remark}

\begin{theorem}\label{con:two-driver}
The two-layer networked sampled-data system with drive-response mode (\ref{sampled-multi-layer},\ref{two_sam}-\ref{two-layer-driver_s}) is controllable if (1) and (2) hold simultaneously:\par
(1) $\forall\eta\in{M}(\theta_{i,j}^1|\Phi^{1,1}_s)$, $\eta\neq\mathbf{0}$, $\eta(\Delta^1\otimes\mathcal{B}^1(h))\neq\mathbf{0}$ for all $i=1,2,...,r^1$, $j=1,2,...,p^1(i)$.\par
(2) $\forall\eta\in{M}(\theta_{i,j}^2|\Phi^{2,2}_s)$, $\eta\neq\mathbf{0}$, and $\forall\xi\in\Xi_{i,j}$, $[\xi(\Delta^1\otimes\mathcal{B}^1(h)),\eta(\Delta^2\otimes\mathcal{B}^2(h))]\neq\mathbf{0}$ for all $i=1,2,...,r^2, j=1,2,...,p^2(i)$, where $\Xi_{i,j}=\{\xi\in\mathbb{C}^{1\times{Nn}}|\xi(\theta_{i,j}^2I_{Nn}-\Phi^{1,1}_s)=\eta\Phi^{2,1}_s\}$.
\end{theorem}

\begin{proof}
According to the PBH criterion, system (\ref{sampled-multi-layer},\ref{two_sam}-\ref{two-layer-driver_s}) is controllable if $\forall\theta\in\sigma(\bm{\Phi}_s),[\theta{I}_{2Nn}-\bm{\Phi}_s,\bm{\Psi}_s]$ is of full row rank.
If system (\ref{sampled-multi-layer},\ref{two_sam}-\ref{two-layer-driver_s}) is uncontrollable, i.e., $\exists\theta\in\sigma(\bm{\Phi}_s)$, and $rank([\theta{I}_{2Nn}-\bm{\Phi}_s,\bm{\Psi}_s])<2Nn$.
If $\theta\in\sigma(\Phi^{2,2}_s)$, there exists some nonzero $[\xi,\eta]$ satisfying
\noindent
$\begin{aligned}
    &[\xi,\eta][I_{2Nn}-\bm{\Phi}_s,\bm{\Psi}_s]\\
    &=[\xi,\eta]\\&\left[ {\begin{array}{c:c}\begin{matrix}
    \theta{I}_{Nn}-\Phi^{1,1}_s & O\\
    -\Phi^{2,1}_s & \theta{I}_{Nn}-\Phi^{2,2}_s
    \end{matrix}&
    \begin{matrix}
    \Delta^1\otimes\mathcal{B}^1(h) & O\\
    O & \Delta^2\otimes\mathcal{B}^2(h)
    \end{matrix}
    \end{array}} \right]\\
    &=\mathbf{0}.
\end{aligned}$
Therefore, $\eta\in{M}(\theta|\Phi^{2,2}_s)$, $\xi(\theta{I}_{Nn}-\Phi^{1,1}_s)=\eta\Phi^{2,1}_s$, and $[\xi(\Delta^1\otimes\mathcal{B}^1(h)),\eta(\Delta^2\otimes\mathcal{B}^2(h))]=\mathbf{0}$. 
If $\eta\neq\mathbf{0}$, condition (2) is contradicted.
If $\eta=\mathbf{0}$, it is easy to find that $\xi\neq\mathbf{0}$, and $\theta$ is also an eigenvalue of $\Phi_s^{1,1}$, so $\xi\in{M}(\theta|\Phi_s^{1,1})$, and $\xi(\Delta^1\otimes{B}^1(h))=\mathbf{0}$, which contradicts condition (1).
Otherwise, $\theta\in\sigma(\Phi^{1,1}_s)$ and $\theta\notin\sigma(\Phi^{2,2}_s)$.
Then there exists a nonzero $\eta$ satisfying:\par
\noindent
$\begin{aligned}
    &[\eta,\mathbf{0}][I_{2Nn}-\bm{\Phi}_s,\bm{\Psi}_s]\\
    &=[\eta,\mathbf{0}]\\&\left[ {\begin{array}{c:c}\begin{matrix}
    \theta{I}_{Nn}-\Phi^{1,1}_s & O\\
    -\Phi^{2,1}_s & \theta{I}_{Nn}-\Phi^{2,2}_s
    \end{matrix}&
    \begin{matrix}
    \Delta^1\otimes\mathcal{B}^1(h) & O\\
    O & \Delta^2\otimes\mathcal{B}^2(h)
    \end{matrix}
    \end{array}} \right]\\
    &=\mathbf{0}.
\end{aligned}$
It indicates that $\eta\in{M}(\theta|\Phi^{1,1}_s)$ and $\eta(\Delta^1\otimes\mathcal{B}^1(h))=\mathbf{0}$, which contradicts condition (1).
The proof is complete.
\end{proof}

\begin{remark}
Condition (1) of Theorem \ref{con:two-driver} is a sufficient controllability condition for the drive layer, but the controllability of the response layer can not be independently verified by the intra-layer condition.
Even if the response layer is uncontrollable itself, the whole multilayer networked sampled-data system can still be controllable due to the inter-layer interactions from the drive layer, which is illustrated in Example \ref{exa_2} in Section \ref{sec:exa_1}.
Condition (2) of Theorem \ref{con:two-driver} indicates that it is possible to alter the controllability of system (\ref{sampled-multi-layer},\ref{two_sam}-\ref{two-layer-driver_s}) by modifying the configuration of the inter-layer couplings.\par
\end{remark}

\begin{corollary}\label{nece}
If $0\notin\sigma(\bm{\Phi}_s)$, the two-layer networked sampled-data system (\ref{sampled-multi-layer},\ref{two_sam}-\ref{two-layer-driver_s}) is controllable if and only if (1) and (2) in Theorem \ref{con:two-driver} hold simultaneously.
\end{corollary}

\begin{proof}
According to the PBH rank condition, if the state matrix $\bm{\Phi}_s$ is non-singular, i.e., $0\notin\sigma(\bm{\Phi}_s)$, system (\ref{sampled-multi-layer},\ref{two_sam}-\ref{two-layer-driver_s}) is controllable if and only if $\forall\theta\in\sigma(\bm{\Phi}_s),[\theta{I}_{2Nn}-\bm{\Phi}_s,\bm{\Psi}_s]$ is of full row rank.
The part of sufficiency has been provided above.
The necessity part is given as follows.\par
If condition (1) is not satisfied, i.e., $\exists\theta\in\sigma(\Phi^{1,1}_s)$ and a corresponding nonzero $\eta\in{M}(\theta|\Phi^{1,1}_s)$, and $\eta(\Delta^1\otimes\mathcal{B}^1(h))=\mathbf{0}$.
Thus, $[\eta,\mathbf{0}][\theta{I}_{2Nn}-\bm{\Phi}_s,\bm{\Psi}_s]=\mathbf{0}$, which indicates that system
(\ref{sampled-multi-layer},\ref{two_sam}-\ref{two-layer-driver_s}) is uncontrollable.
If condition (2) does not hold, i.e., $\exists\theta\in\sigma(\Phi^{2,2}_s)$ and $\eta\in{M}(\theta|\Phi^{2,2}_s),\xi\in\mathbb{C}^{1\times{Nn}}$, satisfying $\xi(\theta{I}_{Nn}-\Phi^{1,1}_s)=\eta\Phi^{2,1}_s$ and $[\xi(\Delta^1\otimes\mathcal{B}^1(h)),\eta(\Delta^2\otimes\mathcal{B}^2(h))]=\mathbf{0}$.
Thus, $[\eta,\xi][I_{2Nn}-\bm{\Phi}_s,\bm{\Psi}_s]=\mathbf{0}$, which suggests that system (\ref{sampled-multi-layer},\ref{two_sam}-\ref{two-layer-driver_s}) is uncontrollable.
\end{proof}

\begin{corollary}\label{coro_topo}
If $0\notin\sigma(\bm{\Phi}_s)$, the two-layer networked sampled-data system (\ref{sampled-multi-layer},\ref{two_sam}-\ref{two-layer-driver_s}) is controllable only if the topology of the first layer $(W^1,\Delta^1)$ is controllable.
\end{corollary}

\begin{proof}
As is shown in Corollary \ref{nece}, if $0\notin\sigma(\bm{\Phi}_s)$, system (\ref{sampled-multi-layer},\ref{two_sam}-\ref{two-layer-driver_s}) is controllable only if (1) and (2) in Theorem \ref{con:two-driver} hold simultaneously.\par
If $(W^1,\Delta^1)$ is not controllable, there exists some $\lambda^1_k\in\sigma(W^1),k\in\{1,...,r^1\}$ and its eigenvector $v^1_k(1)$, which satisfy $v^1_k(1)\Delta^1=\mathbf{0}$.
If the geometric multiplicity of $\lambda_k^1$ is $1$, consider some $\theta_{k,j}^1\in\sigma(E_k^1)$ and its eigenvector $\xi_{k,j}^1(1)$. For $\eta_{k,j}^1(1)\in{M}(\theta_{k,j}^1|\Phi_s^{1,1})$, and $\eta_{k,j}^1(1)(\Delta^1\otimes\mathcal{B}^1(h))=(v^1_k(1)\otimes\xi_{k,j}^1(1))(\Delta^1\otimes\mathcal{B}^1(h))=(v^1_k(1)\Delta^1)\otimes(\xi_{k,j}^1(1)\mathcal{B}^1(h))=\mathbf{0}$, condition (1) of Theorem \ref{con:two-driver} is contradicted thus system (\ref{sampled-multi-layer},\ref{two_sam}-\ref{two-layer-driver_s}) is uncontrollable.
If the geometric multiplicity of $\lambda_k^1$ is $q>1$, it has $q$ linearly independent eigenvectors $v_{k_1}^1(1),v_{k_2}^1(1),...,v_{k_q}^1(1)$, where $k_i\in\{1,2,...,r^1\}$, $i\in\{1,...,q\}$. 
There exists some $v_k=\sum_{i=1}^q a_iv_{k_i}^1(1)$, where $a_i\in\mathbb{C},[a_1,...,a_q]\neq\mathbf{0}$, satisfying $v_k\Delta^1=\mathbf{0}$.
It is obvious that $E_{k_1}^1=...=E_{k_q}^1=e^{A^1h}+\lambda_k^1\mathcal{H}^1(h)$.
Consider some $\theta_{k,j}^1\in\sigma(e^{A^1h}+\lambda_k^1\mathcal{H}^1(h))$ and the corresponding eigenvector $\xi_{k,j}^1(1)$, then one has
$$\begin{aligned}
    &(v_k\otimes\xi^1_{k,j}(1))(\theta_{k,j}^1I_{Nn}-\Phi^{1,1}_s)\\
    =&v_k\otimes\xi^1_{k,j}(1)(\theta_{k,j}^1I_n-e^{A^1h})-\sum_{i=1}^q a_iv_{k_i}^1(1)W^1\otimes\xi_{k,j}^1(1)\mathcal{H}^1(h)\\
    =&v_k\otimes\xi^1_{k,j}(1)(\theta_{k,j}^1I_n-e^{A^1h})-\lambda_k^1\sum_{i=1}^q a_iv_{k_i}^1(1)\otimes\xi_{k,j}^1(1)\mathcal{H}^1(h)\\
    =&v_k\otimes\xi_{k,j}^1(1)(\theta_{k,j}^1I_n-e^{A^1h}-\lambda_k^1\mathcal{H}^1(h))\\
    =&\mathbf{0},
\end{aligned}$$
thus $v_k\otimes{\xi}_{k,j}^1(1)$ is an eigenvector of $\Phi^{1,1}_s$ with respect to $\theta_{k,j}^1$.
Since $(v_k\otimes{\xi}_{k,j}^1(1))(\Delta^1\otimes\mathcal{B}^1(h))=(v_k\Delta^1)\otimes(\xi_{k,j}^1(1)\mathcal{B}^1(h))=\mathbf{0}$, condition (1) of Theorem \ref{con:two-driver} is contradicted, thus system (\ref{sampled-multi-layer},\ref{two_sam}-\ref{two-layer-driver_s}) is uncontrollable.
\end{proof}

\subsection{Homogeneous Situation}

Now discuss the homogeneous case that $A^1=A^2$, $B^1=B^2$, $C^1=C^2$, and $H^1=H^2=P^{2,1}$.
That is, all node systems in the two-layer networked sampled-data system are identical, and the difference among the inner-couplings is ignored.
Then $\bm{\Phi}_s$ and $\bm{\Psi}_s$ in (\ref{sampled-multi-layer}) are
\begin{equation}\label{homo-two-driver_s}
\begin{aligned}
&\bm{\Phi}_s=I_{2N}\otimes{e^{Ah}}+\bar{W}\otimes{\mathcal{H}(h)},\\
&\bm{\Psi}_s=\bar{\Delta}\otimes{\mathcal{B}(h)},
\end{aligned}
\end{equation}
where 
$$\bar{W}=\begin{bmatrix}
W^1 & O\\
D^{2,1} & W^2
\end{bmatrix}, \  \bar{\Delta}=\begin{bmatrix}
\Delta^1 & O\\
O & \Delta^2
\end{bmatrix},$$
and $\mathcal{B}(h)=\int_0^h e^{A\tau} d\tau{B}$, $\mathcal{H}(h)=\int_0^h e^{A\tau} d\tau{HC}$.
Moreover, assume that $\sigma(W^1)\cap\sigma(W^2)=\emptyset$.
Following the method in \cite{hao2019new}, the eigenvalues and corresponding eigenspace of $\bm{\Phi}_s$ in system (\ref{sampled-multi-layer},\ref{homo-two-driver_s}) are derived as follows.

\begin{lemma}\label{lem_homo}
Let $\sigma(W^K)=\{\lambda_1^K,...,\lambda_{r^K}^K\}$, and $\sigma(E^K_i)=\{\theta_{i,1}^K,...,\theta_{i,p(i)}^K\}$, where $E_i^K=e^{Ah}+\lambda_i^K\mathcal{H}(h)$, $i=1,...,r^K$.
Then $\sigma(\Phi^{K,K}_s)=\{\theta_{1,1}^K,...,\theta_{1,p(1)}^K,...,\theta_{r^K,1}^K,...,$ $\theta_{r^K,p(r^K)}^K\}$, $K=1,2$, and $\sigma(\bm{\Phi}_s)=\sigma(\Phi_s^{1,1})\cup\sigma(\Phi_s^{2,2})$.
\end{lemma}

\begin{lemma}\label{eigen_homo}
Assume that the left Jordan chain of $W^K$ with respect to $\lambda_i^K$ is $v_i^K(1),...,v_i^K(\alpha_i^K)$, and the generalized left Jordan chain of $E_i^K$ about $\mathcal{H}(h)$ related to $\theta_{i,j}^K$ is $\xi_{i,j}^K(1),...,\xi_{i,j}^K(\gamma_{i,j}^K)$, where $K=1,2$, $i=1,2,...,r^K$ and $j=1,2,...,p^K(i)$.
The eigenspace of $\bm{\Phi}_s$ with respect to $\theta^K_{i,j}$ is $M(\theta^K_{i,j}|\bm{\Phi}_s)=V(\theta^K_{i,j})$, where
$V(\theta^K_{i,j})=span\{\eta^K_{i,j}(1),\eta^K_{i,j}(2),...,\eta^K_{i,j}(\beta^K_{i,j})\}$, $\beta^K_{i,j}=min\{\alpha^K_i,\gamma^K_{i,j}\}$.
For $K=1$, $\eta^1_{i,j}(1)=[v_i^1(1),\mathbf{0}]\otimes\xi_{i,j}^1(1),\eta^1_{i,j}(2)=[v_i^1(2),\mathbf{0}]\otimes\xi_{i,j}^1(1)+[v_i^1(1),\mathbf{0}]\otimes\xi_{i,j}^1(2),...,\eta^1_{i,j}(\beta_{i,j}^1)=[v_i^1(\beta_{i,j}^1),\mathbf{0}]\otimes\xi^1_{i,j}(1)+...+[v_i^1(1),\mathbf{0}]\otimes\xi^1_{i,j}(\beta^1_{i,j})$.
For $K=2$, $\eta^2_{i,j}(1)=[v_i^{2,1}(1),v_i^2(1)]\otimes\xi_{i,j}^2(1),\eta^2_{i,j}(2)=[v_i^{2,1}(2),v_i^2(2)]\otimes\xi_{i,j}^2(1)+[v_i^{2,1}(1),v_i^2(1)]\otimes\xi_{i,j}^2(2),...,\eta^2_{i,j}(\beta_{i,j}^2)=[v_i^{2,1}(\beta_{i,j}^2),v_i^2(\beta_{i,j}^2)]\otimes\xi^2_{i,j}(1)+...+[v_i^{2,1}(1),v_i^1(1)]\otimes\xi^2_{i,j}(\beta^2_{i,j})$, where $v_i^{2,1}(1)=-v_i^2(1)D^{2,1}(W^1-\lambda_i^2I_N)^{-1}$, and $v_i^{2,1}(k)=(v_i^{2,1}(k-1)-v_i^2(k)D^{2,1})(W^1-\lambda_i^2I_N)^{-1},k=2,...,\beta_{i,j}^2$.
Specially, if $\theta^{K_1}_{i_1,j_1}=...=\theta^{K_q}_{i_q,j_q}$, $q>1$, the eigenspace of $\bm{\Phi}_s$ associated with $\theta^K_{i,j}$ should be the direct sum of all the eigenspace about $\theta^K_{i,j}$, i.e., $M(\theta^K_{i,j}|\bm{\Phi}_s)=\oplus_{l=1}^q V(\theta^{K_l}_{i_l,j_l})$.
\end{lemma}

\begin{proof}
It is obvious that $\sigma(\bar{W})=\sigma(W^1)\cup\sigma(W^2)$.
Define 
$$\bar{T}=\begin{bmatrix}
T^1 & O\\
T^{2,1} & T^2
\end{bmatrix}, 
T^K=\begin{bmatrix}
        v^K_1(\alpha_1^K)\\
        \vdots\\
        v^K_1(1)\\
        \vdots\\
        v^K_{r^K}(\alpha_{r^K}^K)\\
        \vdots\\
        v^K_{r^K}(1)\end{bmatrix},
    T^{2,1}=\begin{bmatrix}
        v^{2,1}_1(\alpha_1^2)\\
        \vdots\\
        v^{2,1}_1(1)\\
        \vdots\\
        v^{2,1}_{r^{2}}(\alpha_{r^2}^2)\\
        \vdots\\
        v^{2,1}_{r^2}(1)\end{bmatrix},$$
$K=1,2$.
Let the invertible matrices $T^1,T^2$ satisfies $T^1W^1(T^1)^{-1}=J^1$, $T^2W^2(T^2)^{-1}=J^2$, and $J^2T^{2,1}-T^{2,1}W^1=T^2D^{2,1}$.
Assume that $\bar{J}=diag\{J^1,J^2\},$ where $J^K$ is the Jordan normal form of $W^K$, and $J^K=diag\{J^K_1,...,J^K_{r^K}\}$.
One has 
$$\begin{aligned}
&(\bar{T}\otimes{I_n})\bm{\Phi}_s(\bar{T}^{-1}\otimes{I_n})\\
=&(\begin{bmatrix}
T^1 & O\\
T^{2,1} & T^2 \end{bmatrix}\otimes{I_n})(I_{2N}\otimes{e}^{Ah}+\bar{W}\otimes\mathcal{H}(h))\\&(\begin{bmatrix}
(T^1)^{-1} & O\\
-(T^2)^{-1}T^{2,1}(T^1)^{-1} & (T^2)^{-1} \end{bmatrix}\otimes{I_n})\\
=&I_{2N}\otimes{e}^{Ah}+\bar{J}\otimes\mathcal{H}(h)\\
=&blockdiag\{F_1^1,...,F^1_{r^1},F_1^2,...,F^2_{r^2}\}.
\end{aligned}$$
where 
$$\begin{aligned}
    F_i^K&=I_{\alpha_i^K}\otimes{e}^{Ah}+J_i^K\otimes\mathcal{H}(h)\\
    &=\left[\begin{matrix}
E_i^K & \mathcal{H}(h) & & \\
         &\ddots &\ddots & \\
         &       & E_i^K & \mathcal{H}(h)\\
         &       &          & E_i^K \end{matrix}\right]_{(\alpha_i^Kn\times{\alpha_i^Kn})}\\
\end{aligned},$$
$i=1,...,r^K,K=1,2.$
The left eigenvectors of $F_i^K$ associated with $\theta_{i,j}^K$ are $e_{\alpha_i^K}\otimes\xi_{i,j}^K(1), e_{\alpha_i^K-1}\otimes\xi_{i,j}^K(1)+e_{\alpha_i^K}\otimes\xi_{i,j}^K(2),...,e_{\alpha_i^K-\beta_{i,j}^K+1}\otimes\xi_{i,j}^K(1)+e_{\alpha_i^K-\beta_{i,j}^K+2}\otimes\xi_{i,j}^K(2)+...+e_{\alpha_i^K}\otimes\xi_{i,j}^K({\beta_{i,j}^K})$, where $\beta_{i,j}^K={min}\{\alpha_i^K,\gamma_{i,j}^K\},i=1,2,...,r^K,j=1,2,...,p_i^K,K=1,2$.
Consider a left eigenvector $\zeta$ of $F_i^K$ associated with $\theta_{i,j}^K$.
It follows that
$$\zeta(\bar{T}_i^K\otimes{I_n})\bm{\Phi}_s=\zeta{F_i^K}(\bar{T}_i^K\otimes{I_n})=\theta_{i,j}^K\zeta(\bar{T}_i^K\otimes{I_n}),$$
which means that $\zeta(\bar{T}_i^K\otimes{I_n})$ is a left eigenvector of $\bm{\Phi}_s$ about $\theta_{i,j}^K$, where $$\bar{T}_i^1=\begin{bmatrix}
                v_i^1({\alpha_i^1}) & 0\\
                \vdots & \vdots\\
                v_i^1({1}) & 0
                \end{bmatrix},
    \bar{T}_i^2=\begin{bmatrix}
                v_i^{2,1}({\alpha_i^2}) & v_i^2({\alpha_i^2})\\
                \vdots & \vdots\\
                v_i^{2,1}({1}) & v_i^2({1})
                \end{bmatrix}.$$
Then the eigenvectors of the state matrix $\bm{\Phi}_s$ of the whole system can be obtained as presented in Lemma \ref{eigen_homo}.
\end{proof}

Based on Lemma \ref{eigen_homo} and PBH rank condition, a sufficient controllability condition for system (\ref{sampled-multi-layer},\ref{homo-two-driver_s}) is given in Theorem \ref{con:two-driver_homo_s}.
The proof is omitted.

\begin{theorem}\label{con:two-driver_homo_s}
The two-layer homogeneous networked sampled-data system with drive-response mode (\ref{sampled-multi-layer},\ref{homo-two-driver_s}) is controllable if $\forall\eta\in{M(\theta^K_{i,j}|\bm{\Phi}_s)}$ and $\eta\neq\mathbf{0}$, $\eta(\bar{\Delta}\otimes\mathcal{B}(h))\neq\mathbf{0}$, for every $K=1,2$, $i=1,...,r^K$ and $j=1,...,p^K(i)$.\par
\end{theorem}

The derivation of Theorem \ref{con:two-driver_homo_s} analyzes the topology of the homogeneous two-layer sampled-data system with drive-response mode, and verifies its controllability by taking it as a general networked sampled-data system. 
Example \ref{exa_1} in Section \ref{sec:exa_2} shows the process of verification.

\begin{remark}
In Example \ref{exa_1} in Section \ref{sec:exa_2}, it can be verified that $[sI_2-A,B]$ is of full row rank for $\forall{s}\in\mathbb{C}$, but the controllability of single node systems is lost after the control sampling, for $rank([sI_2-e^{Ah},\mathcal{B}(h)])=1$ when $s=-23.1407$.
That is, $h=\pi$ is pathological about $A$\cite{chen2012optimal}.
However, the pathological sampling of single node systems is eliminated by the inner-couplings among different nodes, and the whole multilayer networked sampled-data system is still controllable.
\end{remark}

Similar to the proof of Corollary \ref{coro_topo}, the following result in Corollary \ref{coro3} can be obtained for the homogeneous system (\ref{sampled-multi-layer},\ref{homo-two-driver_s}).

\begin{corollary}\label{coro3}
If $0\notin\sigma(\bm{\Phi}_s)$, the two-layer networked sampled-data system (\ref{sampled-multi-layer},\ref{homo-two-driver_s}) is controllable only if the topology $(\bar{W},\bar{\Delta})$ is controllable.
\end{corollary}

\section{Drive-response Mode of Deep Networked Systems}\label{sec:deep}
As the expansion of Section \ref{sec:2_dri}, the drive-response mode in deep networked sampled-data systems is investigated in this section.
Two typical types of inter-layer structure are considered: chain structure (Fig.\ref{fig:chain}) and star structure (Fig. \ref{fig:star}).
In both types, the inter-layer interactions are one-way from the upstream layers to the downstream layers.

\subsection{Chain Structure}
Consider an $M$-layer networked sampled-data system with the chain inter-layer structure, where $\bm{\Phi}_s$ and $\bm{\Psi}_s$ in (\ref{sampled-multi-layer}) are
\begin{equation}\label{chain_layer}
\begin{aligned}
\bm{\Phi}_s&=\begin{bmatrix}
\Phi^{1,1}_s & O & \cdots & O\\
 \Phi^{2,1}_s & \Phi^{2,2}_s & \ddots & \vdots\\
 \vdots & \ddots & \ddots & O\\
 O & \cdots & \Phi^{M,M-1}_s & \Phi^{M,M}_s
\end{bmatrix},\\
\bm{\Psi}_s&=\begin{bmatrix}
\Psi^{1,1}_s & & & \\
  & \Psi^{2,2}_s & & \\
  &  &  \ddots & \\
  &  &  & \Psi^{M,M}_s
\end{bmatrix}.    
\end{aligned}
\end{equation}
It is easy to find $\sigma(\bm{\Phi}_s)=\sigma(\Phi^{1,1}_s)\cup\sigma(\Phi^{2,2}_s)\cup{...}\cup\sigma(\Phi^{M,M}_s)$.
Assume that $\sigma(\bm{\Phi}_s)=\{\theta_1,...,\theta_p\}$.
Two auxiliary definitions of \textit{incompatible eigenvalue set} and \textit{inter-layer coupling chain} are given as follows.

\begin{definition}\label{def:sigma}
If $\theta_i\in\sigma(\Phi^{j,j}_s)$ only when $j=K$, where $i\in\{1,...,p\},K\in\{1,2,...,M\}$, then $\theta_i\in\sigma^K(\bm{\Phi}_s)$.
Otherwise, if $\theta_i\in\sigma(\Phi^{K_1,K_1}_s)\cap...\cap\sigma(\Phi^{K_q,K_q}_s)$, $1\leq{K_1}<...<{K_q}\leq{M}$, then $\theta_i\in\sigma^{K_q}(\bm{\Phi}_s)$.
$\sigma^{K}(\bm{\Phi}_s)$ is called the $K$th incompatible eigenvalue set of $\bm{\Phi}_s$.
Obviously, $\sigma(\bm{\Phi}_s)=\sigma^1(\bm{\Phi}_s)\cup...\cup\sigma^M(\bm{\Phi}_s)$.
\end{definition}

\begin{definition}
A series of vectors $\xi_1,\xi_2,...,\xi_{K-1},\xi_K$ form an inter-layer coupling chain for $\{\Phi^{1,1}_s,...,\Phi^{K,K}_s\}$ associated with $\{\Phi^{2,1}_s,...,\Phi^{K,K-1}_s\}$ about $\theta_i\in\sigma^K(\bm{\Phi}_s)$, where $i\in\{1,...,p\},K\in\{2,...,M\}$, if:
\noindent
$$\begin{aligned}
    &\xi_K(\theta_i{I}_{Nn}-\Phi^{K,K}_s)=\mathbf{0},\\
    \text{and} \ \ 
&\xi_j(\theta_i{I_{Nn}}-\Phi^{j,j}_s)=\xi_{j+1}\Phi^{j+1,j}_s,\ j\in\{1,...,K-1\}.
\end{aligned}$$
\end{definition}

Then the eigenspace of system (\ref{sampled-multi-layer},\ref{chain_layer}) can be derived, as is shown in Lemma \ref{eigen:chain_layer}.

\begin{lemma}\label{eigen:chain_layer}
If $\theta_i\in\sigma^K(\bm{\Phi}_s)$, where $i\in\{1,...,p\},K\in\{1,...,M\}$, then $M(\theta_i|\bm{\Phi}_s)=span\{\eta^i_1,\eta^i_2,...,\eta^i_{\gamma_i}\}$.\par
(1) If $K=1$: Let $M(\theta_i|\Phi^{1,1}_s)=span\{\xi^i_1,...,\xi^i_{\gamma_i}\}$, then $\eta^i_j=[\xi^i_j,\mathbf{0},...,\mathbf{0}]$, where $j=1,...,\gamma_i$;\par
(2) If $1<K\leq{M}$: $\eta^i_j=[\xi^i_j(1),...,\xi^i_j(K-1),$ $\xi^i_j(K),...,\mathbf{0}]$, where $\xi^i_j(1),...,\xi^i_j(K-1),\xi^i_j(K)$ is the $j$th inter-layer coupling chain for $\{\Phi^{1,1},...,\Phi^{K,K}\}$ associated with $\{\Phi^{2,1}_s,...,\Phi^{K,K-1}_s\}$ about $\theta_i$, $j=1,...,\gamma_i$.\par
\end{lemma}

Based on Lemma \ref{eigen:chain_layer} and the PBH rank condition, Theorem \ref{con:deep} is obtained to verify the controllability of system (\ref{sampled-multi-layer},\ref{chain_layer}).
The proof is omitted.

\begin{theorem}\label{con:deep}
The deep networked sampled-data system with chain inter-layer structure (\ref{sampled-multi-layer},\ref{chain_layer}) is controllable if for every $i=1,...,p$, $\forall{\eta}\in{M}(\theta_i|\bm{\Phi}_s)$ and $\eta\neq\mathbf{0}$, $\eta\bm{\Psi}_s\neq\mathbf{0}$.
\end{theorem}

\subsection{Star Structure}
Consider an $M$-layer networked sampled-data system with the star inter-layer structure, where $\bm{\Phi}_s$ and $\bm{\Psi}_s$ in (\ref{sampled-multi-layer}) are:
\begin{equation}\label{star_layer}
\begin{aligned}
\bm{\Phi}_s&=\begin{bmatrix}
\Phi^{1,1}_s & O & \cdots & O\\
 \Phi^{2,1}_s & \Phi^{2,2}_s & \ddots & \vdots\\
 \vdots &  & \ddots & O\\
 \Phi^{M,1}_s & O &  & \Phi^{M,M}_s
\end{bmatrix},\\
\bm{\Psi}_s&=\begin{bmatrix}
\Psi^{1,1}_s & & & \\
  & \Psi^{2,2}_s & & \\
  &  &  \ddots & \\
  &  &  & \Psi^{M,M}_s
\end{bmatrix}.
\end{aligned}
\end{equation}
Still, $\sigma(\bm{\Phi}_s)=\sigma(\Phi^{1,1}_s)\cup\sigma(\Phi^{2,2}_s)\cup{...}\cup\sigma(\Phi^{M,M}_s)$.
Assume that $\sigma(\bm{\Phi}_s)=\{\theta_1,...,\theta_p\}$.
Similar to the case of chain structure, two definitions of \textit{center-incompatible eigenvalue set} and \textit{inter-layer coupling group} are given as follows.

\begin{definition}\label{def:sigma2}
If $\theta_i\in\sigma(\Phi^{1,1}_s)\cap\sigma(\Phi^{K,K}_s)$, where $i\in\{1,...,p\},K\in\{2,...,M\}$, then $\theta_i\in\sigma^K(\bm{\Phi}_s)$.
Otherwise, if $\theta_i\in\sigma(\Phi^{K,K}_s),K\in\{1,2,...,M\}$, then $\theta_i\in\sigma^K(\bm{\Phi}_s)$.
$\sigma^{K}(\bm{\Phi}_s)$ is called the $K$th center-incompatible eigenvalue set of $\bm{\Phi}_s$.
Obviously, $\sigma(\bm{\Phi}_s)=\sigma^1(\bm{\Phi}_s)\cup...\cup\sigma^M(\bm{\Phi}_s)$.
\end{definition}

\begin{figure}[htb]
\centering
\subfloat[]{\includegraphics[width=1.4in]{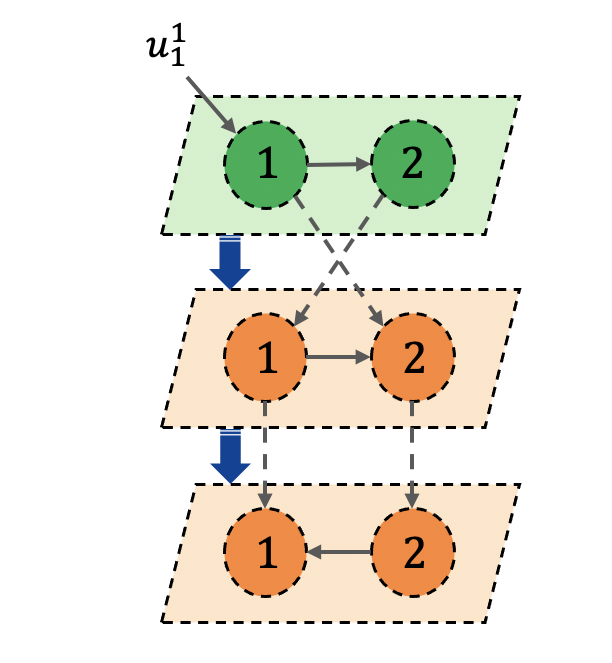}%
\label{fig:chain}}
\hfil
\subfloat[]{\includegraphics[width=1.95in]{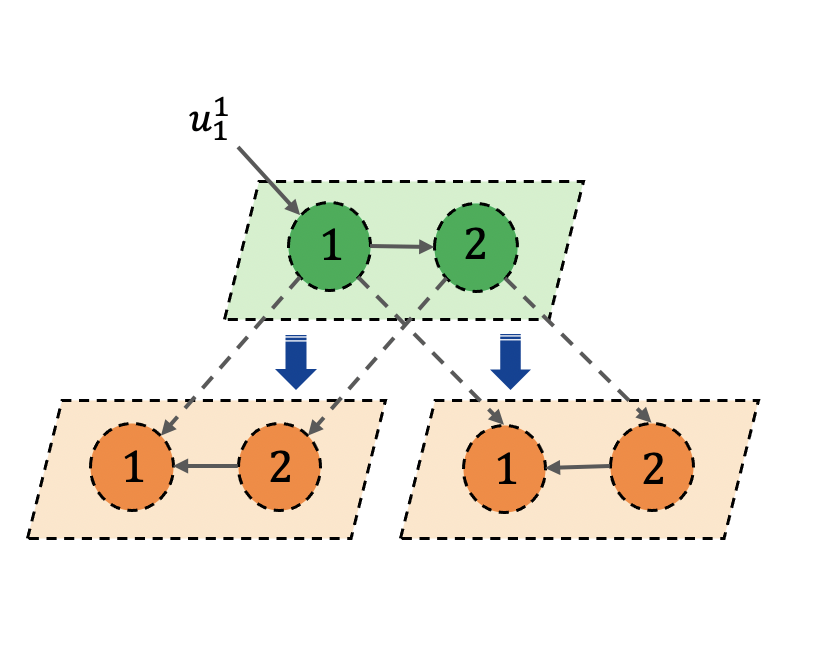}%
\label{fig:star}}
\caption{Two types of deep networked systems. (a) Chain structure. (b) Star structure.}
\label{fig:deep}
\end{figure}

\begin{definition}
A series of vectors $\xi_1,\xi_{K_1},\xi_{K_2},...,\xi_{K_l}$ constitute an inter-layer coupling group certering in $\xi_1$ for $\{\Phi^{1,1}_s,\Phi^{K_1,K_1}_s...,\Phi^{K_l,K_l}_s\}$ associated with $\{\Phi^{K_1,1}_s,...,\Phi^{K_l,1}_s\}$ about $\theta_i\in\sigma^{K_1}(\bm{\Phi}_s)\cup...\cup\sigma^{K_l}(\bm{\Phi}_s)$, where $i\in\{1,...,p\},1<K_1<...<K_l\leq{M}$, if:
\noindent
$$\begin{aligned}
    &\xi_{K_j}\in{M}(\theta_i|\Phi^{K_j,K_j}_s),j=1,...,l\\
    \text{and} \ \
&\xi_1(\theta_i{I}_{Nn}-\Phi^{1,1}_s)=\xi_{K_1}\Phi^{K_1,1}_s+...+\xi_{K_l}\Phi^{K_l,1}_s.
\end{aligned}$$
\end{definition}

Then the eigenspace of system (\ref{sampled-multi-layer},\ref{star_layer}) can be derived, as is shown in Lemma \ref{eigen_star_layer}.

\begin{lemma}\label{eigen_star_layer}
(1) If $\theta_i\in\sigma^1(\bm{\Phi}_s)$: Let $M(\theta_i|\Phi^{1,1}_s)=span\{\xi^i_1,...,\xi^i_{\gamma_i}\}$, then $M(\theta_i|\bm{\Phi}_s)=span\{\eta^i_1,...,\eta^i_{\gamma_i}\}$, where $\eta^i_j=[\xi^i_j,\mathbf{0},...,\mathbf{0}],j=1,...,\gamma_i$;\par
(2) Otherwise, if $\theta_i\in\sigma^{K_1}(\bm{\Phi}_s)\cap...\cap\sigma^{K_l}(\bm{\Phi}_s)$, $1<K_1<...<K_l\leq{M}$, then $M(\theta_i|\bm{\Phi}_s)=span\{\eta^i_1,...,\eta^i_{\zeta_i}\}$, where $\eta^i_j=[\xi^i_j(1),\mathbf{0},...,\xi^i_j(K_1),\mathbf{0},...,\xi^i_j(K_l),...,\mathbf{0}]$, $\xi^i_j(K_q)\in{M}(\theta_i|\Phi^{K_q,K_q}_s),q=1,...,l$,
and $\xi^i_j(1),\xi^i_j(K_1),...,\xi^i_j(K_l)$ is the $j$th inter-layer coupling group for $\{\Phi^{1,1}_s,\Phi^{K_1,K_1}_s...,\Phi^{K_l,K_l}_s\}$ associated with $\{\Phi^{K_1,1}_s,...,\Phi^{K_l,1}_s\}$ about $\theta_i$ centering in $\xi^i_j(1)$, $j=1,...,\zeta_i$.
\end{lemma}

Based on Lemma \ref{eigen_star_layer}, the controllability of system (\ref{sampled-multi-layer},\ref{star_layer}) can be verified by Theorem \ref{con:deep_2}.
The proof is omitted.

\begin{theorem}\label{con:deep_2}
The deep networked sampled-data system with star inter-layer structure (\ref{sampled-multi-layer},\ref{star_layer}) is controllable if for every $i=1,...,p$, $\forall{\eta}\in{M}(\theta_i|\bm{\Phi}_s)$ and $\eta\neq\mathbf{0}$, $\eta\bm{\Psi}_s\neq\mathbf{0}$.
\end{theorem}

\begin{remark}
Similar to the issue that Corollary \ref{nece} mentions above, if the state matrix $\bm{\Phi}_s$ of the deep networked sampled-data system with star (\ref{sampled-multi-layer},\ref{star_layer}) or chain (\ref{sampled-multi-layer},\ref{chain_layer}) inter-layer structure is non-singular, the condition in Theorem \ref{con:deep} and \ref{con:deep_2} become necessary and sufficient.
The discussion and proof is omitted here.
\end{remark}

\section{Multi-rate Sampling Patterns}\label{sec:multi}
Apart from the single-rate-sampled-data systems discussed above, the sampling periods on different channels can be different in multi-rate-sampled-data systems.
Consider two different sampling periods, $lh$ and $h,$ where $l\in\mathbb{N}^+$.
By simple algebraic transformation, the multi-rate-sampled-data system can be written in a compact form with interval $lh$:
\begin{equation}\label{mul_rate_lh}
    \mathbf{X}((kl+l)h)=\tilde{\bm{\Phi}}_s\mathbf{X}(klh)+\tilde{\bm{\Psi}}_s\tilde{\mathbf{U}}(klh),
\end{equation}
where
\begin{equation}\label{mul_rate_lh_detail}
    \tilde{\bm{\Phi}}_s=\begin{bmatrix}
    \tilde{\Phi}_s^{1,1} & O\\
    \tilde{\Phi}_s^{2,1} & \tilde{\Phi}_s^{2,2}
    \end{bmatrix},
    \tilde{\bm{\Psi}}_s=\begin{bmatrix}
    \tilde{\Psi}_s^{1,1} & O\\
    O & \tilde{\Psi}_s^{2,2}
    \end{bmatrix}.
\end{equation}

As is shown in Example \ref{exa_3}-\ref{exa_4} in Section \ref{sec:exa_34}, the change of the sampling rate on some channels may lay a positive or negative effect on the controllability of the whole multilayer networked sampled-data system.
However, it will cause a large computation burden if each time the sampling rate is modified, the system matrix and input matrix of the whole system are required to be recalculated, as well as the re-verification of the rank criterion.
To solve this problem, in this section, three typical multi-rate sampling patterns (Slow Inter-layer, Multi-scale, and Fast Control) are studied and some controllability conditions are given based on the original single-rate multilayer sampled-data networked system, which are easier-to-verify.

\begin{figure*}[!t]
\centering
\subfloat[]{\includegraphics[width=1.5in]{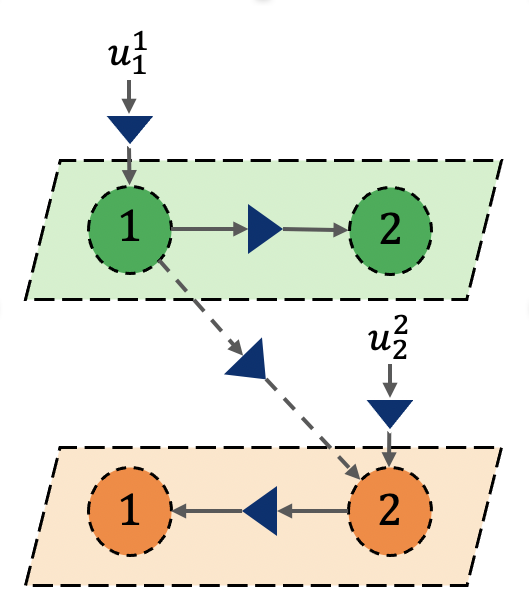}%
\label{fig:single}}
\hfil
\subfloat[]{\includegraphics[width=1.5in]{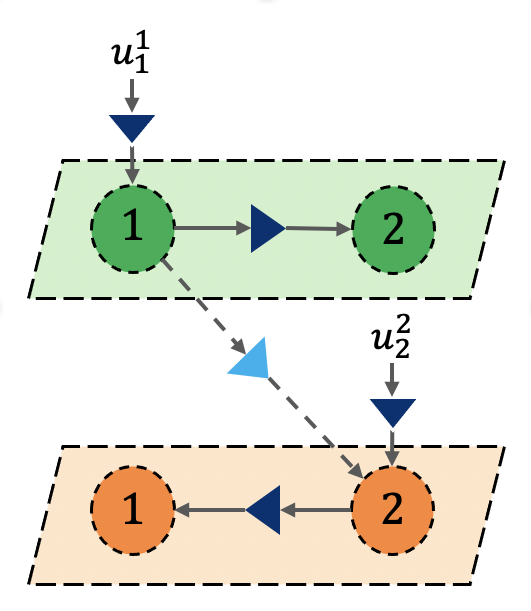}%
\label{fig:slow_inter}}
\hfil
\subfloat[]{\includegraphics[width=1.5in]{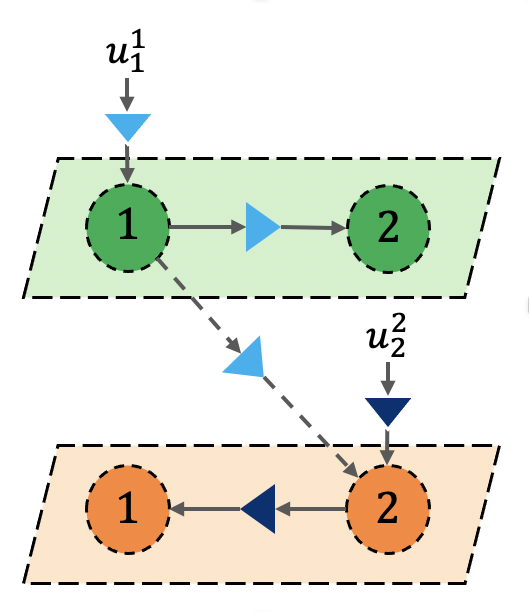}%
\label{fig:multi-scale}}
\hfil
\subfloat[]{\includegraphics[width=1.5in]{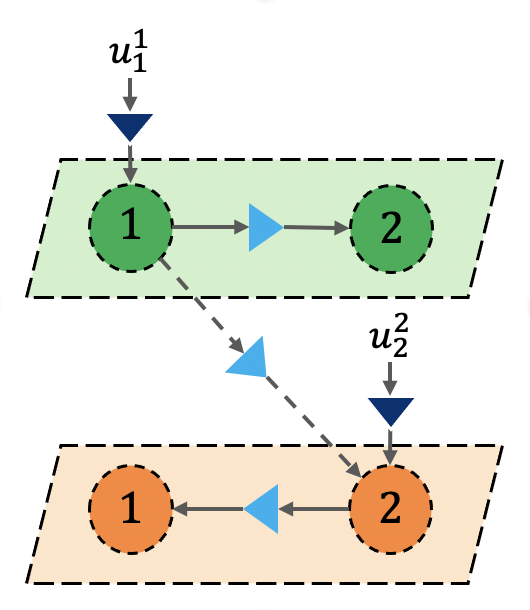}%
\label{fig:fast-con}}
\caption{Single-rate sampling pattern (a) and three multi-rate sampling patterns in the two-layer networked sampled-data system with drive-response mode: (b) Slow Inter-layer, (c) Multi-scale, and (d) Fast Control. The sky-blue (light gray) and navy-blue (dark gray) triangles represent samplers of interval $lh$ and $h$, respectively.}
\label{fig_sim}
\end{figure*}

\subsection{Slow Inter-layer Sampling Pattern}
Assume that the sampling on the inter-layer transmission channels are slower than that on other channels, as is shown in Fig.\ref{fig:slow_inter}, where the sampling period on the inter-layer transmission channels is $lh$, while the sampling period on other channels is $h$.
Perform the Slow Inter-layer sampling pattern, and the dynamics of the multi-rate sampled-data system can be described as:
$$\begin{aligned}
\dot{X}^1&(t)=(I_{N}\otimes{A^1})X^1(t)+(W^1\otimes{H^1C^1})X^1((kl+r)h)\\&+(\Delta^1\otimes{B^1})U^1((kl+r)h),\\
\dot{X}^2&(t)=(I_{N}\otimes{A^2})X^2(t)+(W^2\otimes{H^2C^2})X^2((kl+r)h)\\&+(\Delta^2\otimes{B^2})U^2((kl+r)h)+(D^{2,1}\otimes{P^{2,1}C^1})X^1(klh),\\
\end{aligned}$$
where $t\in((kl+r)h,(kl+r+1)h]$, $r=0,1,...,l-1$ and $k\in\mathbb{N}$.
Thus, in the multi-rate sampled-data system (\ref{mul_rate_lh},\ref{mul_rate_lh_detail}):
\begin{equation}\label{eq:slow-inter}
\begin{aligned}
    &\tilde{\Phi}_s^{1,1}=(\Phi^{1,1}_s)^l, \ \tilde{\Phi}_s^{2,2}=(\Phi^{2,2}_s)^l,\\
    &\tilde{\Phi}_s^{2,1}=((\Phi_s^{2,2})^{l-1}+...+\Phi_s^{2,2}+I_{Nn})\Phi_s^{2,1},\\
    &\tilde{\Psi}_s^{1,1}=[(\Phi^{1,1}_s)^{l-1}\Psi_s^{1,1},...,\Phi_s^{1,1}\Psi_s^{1,1},\Psi_s^{1,1}],\\
    &\tilde{\Psi}_s^{2,2}=[(\Phi^{2,2}_s)^{l-1}\Psi_s^{2,2},...,\Phi_s^{2,2}\Psi_s^{2,2},\Psi_s^{2,2}],
\end{aligned}
\end{equation}
Note that $\Phi_s^{K,K},\Psi_s^{K,K},K=1,2$ and $\Phi_s^{2,1}$ are the same as that in equation (\ref{two-layer-driver_s}).
$\tilde{\mathbf{U}}(klh)=[(\tilde{\mathbf{U}}^1(klh))^\top,(\tilde{\mathbf{U}}^2(klh))^\top]^\top$, where $\tilde{\mathbf{U}}^K(klh)=[(U^K(klh))^\top,(U^K((kl+1)h))^\top,...,$ $(U^K((kl+l-1)h))^\top]^\top,K=1,2$. 

Since system (\ref{mul_rate_lh}-\ref{eq:slow-inter}) is also a two-layer networked sampled-data system with the drive-response mode, a controllability condition can be developed from Theorem \ref{con:two-driver} as follows.

\begin{corollary}\label{con:slow-inter}
The two-layer networked sampled-data system with Slow Inter-layer sampling pattern (\ref{mul_rate_lh}-\ref{eq:slow-inter}) is controllable if (1) and (2) hold simultaneously:\par
(1) $\forall\eta\in{M}((\theta_{i,j}^1)^l|(\Phi_s^{1,1})^l)$, $\eta\neq\mathbf{0}$, $\eta(\Delta^1\otimes\mathcal{B}^1(h))\neq\mathbf{0}$ for every $i=1,2,...,r^1$, $j=1,2,...,p^1(i)$.\par
(2) $\forall\eta\in{M}((\theta_{i,j}^2)^l|(\Phi_s^{2,2})^l)$, $\eta\neq\mathbf{0}$, and $\forall\xi\in\tilde{\Xi}_{i,j}$, $[\xi\tilde{\Psi}_s^{1,1},\eta(\Delta^2\otimes\mathcal{B}^2(h))]\neq\mathbf{0}$ for every $i=1,2,...,r^2$, $j=1,2,...,p^2(i)$, where $\tilde{\Xi}_{i,j}=\{\xi\in\mathbb{C}^{1\times{Nn}}|\xi((\theta_{i,j}^2)^lI_{Nn}-(\Phi_s^{1,1})^l)=\eta\tilde{\Phi}_s^{2,1}\}$.
\end{corollary}

\begin{proof}
According to the spectral mapping theorem, if $\sigma(\Phi_s^{K,K})=\{\theta_{1,1}^K,...,\theta_{1,p^K(1)}^K,...,\theta_{r^K,1}^K,...,\theta_{r^K,p^K(r^K)}^K\},$ then $\sigma((\Phi_s^{K,K})^l)=\{(\theta_{1,1}^K)^l,...,(\theta_{1,p^K(1)}^K)^l,...,(\theta_{r^K,1}^K)^l,...,$ $(\theta_{r^K,p^K(r^K)}^K)^l\}$, and $M(\theta_{i,j}^K|\Phi_s^{K,K})\subset{M}((\theta_{i,j}^K)^l|(\Phi_s^{K,K})^l)$, $K=1,2$.
It is obvious that $\sigma(\tilde{\bm{\Phi}}_s)=\sigma((\Phi_s^{1,1})^l)\cup\sigma((\Phi_s^{2,2})^l)$.
From Theorem \ref{con:two-driver}, system (\ref{mul_rate_lh}-\ref{eq:slow-inter}) is controllable if the following conditions hold simultaneously:\par
(1) $\forall\eta\in{M}((\theta_{i,j}^1)^l|(\Phi_s^{1,1})^l)$, $\eta\neq\mathbf{0}$, $\eta\tilde{\Psi}^{1,1}_s\neq\mathbf{0}$ for every $i=1,2,...,r^1$, $j=1,2,...,p^1(i)$.\par
(2) $\forall\eta\in{M}((\theta_{i,j}^2)^l|(\Phi_s^{2,2})^l)$, $\eta\neq\mathbf{0}$, and $\forall\xi\in\tilde{\Xi}_{i,j},[\xi\tilde{\Psi}^{1,1}_s,\eta\tilde{\Psi}^{2,2}_s]\neq\mathbf{0}$ for every $i=1,2,...,r^2, j=1,2,...,p^2(i)$, where $\tilde{\Xi}_{i,j}=\{\xi\in\mathbb{C}^{1\times{Nn}}|\xi((\theta_{i,j}^2)^lI_{Nn}-(\Phi_s^{1,1})^l)=\eta\tilde{\Phi}_s^{2,1}\}$.\par
If $\eta\in{M}(\theta_{i,j}^K|\Phi_s^{K,K}),$ it follows that
\begin{equation}\label{equ:mul_1}
    \begin{aligned}
        \eta\tilde{\Phi}_s^{K,K}&=\eta[(\Phi^{K,K}_s)^{l-1}\Psi_s^{K,K},...,\Phi^{K,K}_s\Psi_s^{K,K},\Psi_s^{K,K}]\\&=[(\theta_{i,j}^K)^{l-1}\eta\Psi^{K,K}_s,...,\theta_{i,j}^K\eta\Psi^{K,K}_s,\eta\Psi^{K,K}_s].
    \end{aligned}
\end{equation}
Otherwise, if $(\theta_{i_1,j_1}^K)^l=(\theta_{i_2,j_2}^K)^l=...=(\theta_{i_q,j_q}^K)^l=\theta$, where $i_k\in\{1,...,r^K\}$ and $j_k\in\{1,...,p_{i_k}^K\}$, $k=1,...,q,q>1$, then $M(\theta|(\Phi_s^{K,K})^l)=\oplus_{k=1}^qM(\theta_{i_k,j_k}^K)|\Phi_s^{K,K}$.
Consider $\eta=\sum_{k=1}^q a_k\eta_k$, with $[a_1,...,a_q]\neq\mathbf{0}$, $a_k\in\mathbb{C}$, $\eta_k\in{M(\theta_{i_k,j_k}^K)|\Phi_s^{K,K}})$, then
\begin{equation}
    \begin{aligned}
        &\eta\tilde{\Psi}_s^{K,K}=\eta[(\Phi^{K,K}_s)^{l-1}\Psi_s^{K,K},...,\Psi_s^{K,K}]\\
        &=[\sum_{k=1}^q a_k(\theta_{i_k,j_k}^K)^{l-1}\eta_k\Psi^{K,K}_s,...,\sum_{k=1}^q a_k\eta_k\Psi^{K,K}_s].
    \end{aligned}
\end{equation}
Overall, it is easy to see that if $\eta\Phi^{K,K}_s=\eta(\Delta^K\otimes{\mathcal{B}^K(h)})\neq\mathbf{0}$, then $\eta\tilde{\Phi}_s^{K,K}\neq\mathbf{0}$, $K=1,2$.
So far Corollary \ref{con:slow-inter} is proved.
\end{proof}

\subsection{Multi-scale Sampling Pattern}

Assume that the timescale of different layers are different from each other, i.e., the sampling rates of the ZOHs in different layers are different\cite{posfai2016controllability}, as is shown in Fig.\ref{fig:multi-scale}.
Perform the Multi-scale sampling pattern: Let the sampling period of the drive layer be $lh$, and the sampling period of the response layer be $h$.
The dynamics of the multi-scale sampled-data system can be described as:
$$\begin{aligned}
\dot{X}^1&(t)=(I_{N}\otimes{A^1})X^1(t)+(W^1\otimes{H^1C^1})X^1(klh)\\&+(\Delta^1\otimes{B^1})U^1(klh)\\
\dot{X}^2&(t)=(I_{N}\otimes{A^2})X^2(t)+(W^2\otimes{H^2C^2})X^2((kl+r)h)\\&+(\Delta^2\otimes{B^2})U^2((kl+r)h)+(D^{2,1}\otimes{P^{2,1}C^1})X^1(klh),\\
\end{aligned}$$
where $t\in((kl+r)h,(kl+r+1)h]$, $r=0,1,...,l-1$ and $k\in\mathbb{N}$.
Thus, in the multi-rate sampled-data system (\ref{mul_rate_lh},\ref{mul_rate_lh_detail}):
\begin{equation}\label{eq:multiscale}
\begin{aligned}
    &\tilde{\Phi}_s^{1,1}=I_N\otimes{e^{A^1lh}+W^1\otimes{\mathcal{H}^1(lh)}},\\
    &\tilde{\Phi}_s^{2,1}=((\Phi_s^{2,2})^{l-1}+...+\Phi_s^{2,2}+I_{Nn})\Phi_s^{2,1},\\
    &\tilde{\Phi}_s^{2,2}=(\Phi^{2,2}_s)^l, \ \tilde{\Psi}_s^{1,1}=\Delta^1\otimes{\mathcal{B}^1(lh)},\\
    &\tilde{\Psi}_s^{2,2}=[(\Phi^{2,2}_s)^{l-1}\Psi_s^{2,2},...,\Phi_s^{2,2}\Psi_s^{2,2},\Psi_s^{2,2}].
\end{aligned}
\end{equation}
Note that $\Phi_s^{2,2},\Psi_s^{2,2}$ and $\Phi_s^{2,1}$ are the same as that in equation (\ref{two-layer-driver_s}).
$\tilde{\mathbf{U}}(klh)=[(U^1(klh))^\top,(\tilde{\mathbf{U}}^2(klh))^\top]^\top,$ with $\tilde{\mathbf{U}}^2(klh)=[(U^2(klh))^\top,(U^2((kl+1)h))^\top,...,(U^2((kl+l-1)h))^\top]^\top$.

Let the eigenvalues and corresponding eigenspace of $\tilde{\Phi}_s^{1,1}$ be $\tilde{\theta}^1_{i,j}$ and $M(\tilde{\theta}^1_{i,j}|\tilde{\Phi}_s^{1,1})$, $i=1,...,r^1,j=1,...,\tilde{p}^1(i)$, respectively.
Similar to the proof of Corollary \ref{con:slow-inter},
the following controllability condition can be developed from Theorem \ref{con:two-driver} for system (\ref{mul_rate_lh}-\ref{mul_rate_lh_detail},\ref{eq:multiscale}).

\begin{corollary}\label{con:multiscale}
The two-layer networked sampled-data system with Multi-scale sampling pattern (\ref{mul_rate_lh}-\ref{mul_rate_lh_detail},\ref{eq:multiscale}) is controllable if (1) and (2) hold simultaneously:\par
(1) $\forall\eta\in{M}(\tilde{\theta}_{i,j}^1|\tilde{\Phi}_s^{1,1})$, $\eta\neq\mathbf{0}$, $\eta(\Delta^1\otimes\mathcal{B}^1(lh))\neq\mathbf{0}$ for every $i=1,2,...,r^1$, $j=1,2,...,\tilde{p}^1(i)$.\par
(2) $\forall\eta\in{M}((\theta_{i,j}^2)^l|(\Phi_s^{2,2})^l)$, $\eta\neq\mathbf{0}$, and $\forall\xi\in\tilde{\Xi}_{i,j}$, $[\xi(\Delta^1\otimes\mathcal{B}^1(lh)),\eta(\Delta^2\otimes\mathcal{B}^2(h))]\neq\mathbf{0}$ for every $i=1,2,...,r^2$, $j=1,2,...,p^2(i)$, where $\tilde{\Xi}_{i,j}=\{\xi\in\mathbb{C}^{1\times{Nn}}|\xi((\theta_{i,j}^2)^lI_{Nn}-\tilde{\Phi}_s^{1,1})=\eta\tilde{\Phi}_s^{2,1}\}$.
\end{corollary}

\subsection{Fast Control Sampling Pattern}
\label{subsec:fast_contr}

Assume that the sampling on the control channels are faster than that on other channels, as is shown in Fig.\ref{fig:fast-con}, where the sampling period on the the control channels is $h$, while the sampling period on other channels is $lh$.
Perform the Fast Control sampling pattern and the dynamics of the multi-rate sampled-data system can be described as:
$$\begin{aligned}
\dot{X}^1&(t)=(I_{N}\otimes{A^1})X^1(t)+(W^1\otimes{H^1C^1})X^1(klh)\\&+(\Delta^1\otimes{B^1})U^1((kl+r)h)\\
\dot{X}^2&(t)=(I_{N}\otimes{A^2})X^2(t)+(W^2\otimes{H^2C^2})X^2(klh)\\&+(\Delta^2\otimes{B^2})U^2((kl+r)h)+(D^{2,1}\otimes{P^{2,1}C^1})X^1(klh),\\
\end{aligned}$$
where $t\in((kl+r)h,(kl+r+1)h]$, $r=0,1,...,l-1$ and $k\in\mathbb{N}$.
Thus in the multi-rate sampled-data system (\ref{mul_rate_lh},\ref{mul_rate_lh_detail}):
\begin{equation}\label{eq:fast-con}
\begin{aligned}
    \tilde{\Phi}_s^{1,1}=&I_N\otimes{e^{A^1lh}+W^1\otimes{\mathcal{H}^1(lh)}},\\
    \tilde{\Phi}_s^{2,2}=&I_N\otimes{e^{A^2lh}+W^2\otimes{\mathcal{H}^2(lh)}},\\
    \tilde{\Phi}_s^{2,1}=&D^{2,1}\otimes{\mathcal{P}^{2,1}(lh)},\\
    \tilde{\Psi}_s^{1,1}=&[\Delta^1\otimes{(\mathcal{B}^1(lh)-\mathcal{B}^1((l-1)h))},...,\\
    &\Delta^1\otimes{(\mathcal{B}^1(2h)-\mathcal{B}^1(h))},\Delta^1\otimes\mathcal{B}^1(h)],\\
    \tilde{\Psi}_s^{2,2}=&[\Delta^2\otimes{(\mathcal{B}^2(lh)-\mathcal{B}^2((l-1)h))},...,\\
    &\Delta^2\otimes{(\mathcal{B}^2(2h)-\mathcal{B}^2(h))},\Delta^2\otimes\mathcal{B}^2(h)],
\end{aligned}
\end{equation}
$\tilde{\mathbf{U}}(klh)=[(\tilde{\mathbf{U}}^1(klh))^\top,(\tilde{\mathbf{U}}^2(klh))^\top]^\top,$ with $\tilde{\mathbf{U}}^i(klh)=[(U^i(klh))^\top,(U^i((kl+1)h))^\top,...,(U^i((kl+l-1)h))^\top]^\top$, $i=1,2$.

Let the eigenvalues and corresponding eigenspace of $\tilde{\Phi}_s^{K,K}$ be $\tilde{\theta}^K_{i,j}$ and $M(\tilde{\theta}^K_{i,j}|\tilde{\Phi}_s^{K,K})$, $i=1,...,r^K,j=1,...,\tilde{p}^K(i)$, respectively, $K=1,2$.
Similar to the proof of Theorem \ref{con:two-driver}, the following controllability condition can be developed from Theorem \ref{con:two-driver} for system (\ref{mul_rate_lh}-\ref{mul_rate_lh_detail},\ref{eq:fast-con}).

\begin{corollary}\label{con:fast-con}
The two-layer networked sampled-data system with Fast Control sampling pattern (\ref{mul_rate_lh}-\ref{mul_rate_lh_detail},\ref{eq:fast-con}) is controllable if (1) and (2) hold simultaneously:\par
(1) $\forall\eta\in{M}(\tilde{\theta}_{i,j}^1|\tilde{\Phi}_s^{1,1})$, $\eta\neq\mathbf{0}$, $\eta[\Delta^1\otimes\mathcal{B}^1(lh),...,\Delta^1\otimes\mathcal{B}^1(h)]\neq\mathbf{0}$ for every $i=1,2,...,r^1$, $j=1,2,...,\tilde{p}^1(i)$.\par
(2) $\forall\eta\in{M}(\tilde{\theta}_{i,j}^2|\tilde{\Phi}_s^{2,2})$, $\eta\neq\mathbf{0}$, and $\forall\xi\in\tilde{\Xi}_{i,j}$, $[\xi[\Delta^1\otimes\mathcal{B}^1(lh),...,\Delta^1\otimes\mathcal{B}^1(h)],\eta[\Delta^2\otimes\mathcal{B}^2(lh),...,\Delta^2\otimes\mathcal{B}^2(h)]]\neq\mathbf{0}$ for every $i=1,2,...,r^2$, $j=1,2,...,\tilde{p}^2(i)$, where $\tilde{\Xi}_{i,j}=\{\xi\in\mathbb{C}^{1\times{Nn}}|\xi(\tilde{\theta}_{i,j}^2I_{Nn}-\tilde{\Phi}_s^{1,1})=\eta\tilde{\Phi}_s^{2,1}\}$.
\end{corollary}

\section{Simulated Examples}\label{sec:exa}

\subsection{Control of Response-layer by Inter-layer Coupling}\label{sec:exa_1}
\begin{example}\label{exa_2}
Consider the following two-layer networked sampled-data system with drive-response mode, where both layers are chain networks consisting of two nodes, $w^1_{21}=w^2_{21}=1,\delta^1_1=\delta^2_2=1$.
Let sampling period $h=0.1$, and
$$\begin{aligned}
&P^{2,1}=A^1=\left[\begin{array}{cc}
    1 & 1 \\
    0 & 1
\end{array}\right],
B^1=\left[\begin{array}{cc}
    2 & 0 \\
    0 & 1
\end{array}\right],
C^1=\left[\begin{array}{cc}
    2 & 1 \\
    0 & 1
\end{array}\right],\\
&A^2=\left[\begin{array}{cc}
    1 & 0 \\
    1 & 1
\end{array}\right],
B^2=\left[\begin{array}{cc}
    2 & 1 \\
    0 & 2
\end{array}\right],
C^2=\left[\begin{array}{cc}
    1 & 0 \\
    0 & 2
\end{array}\right],\\
&H^1=\left[\begin{array}{cc}
    1 & 1 \\
    1 & 0
\end{array}\right],
H^2=\left[\begin{array}{cc}
    0 & 1 \\
    1 & 0
\end{array}\right],
D^{2,1}=\left[\begin{array}{cc}
    1 & 0 \\
    0 & 0
\end{array}\right].
\end{aligned}$$
One can calculate that $\sigma(W^1)=\sigma(W^2)=\{0\}$.
The corresponding generalized eigenvectors are $v^1_1(1)=v^2_1(1)=[1,0],v^1_1(2)=v^2_1(2)=[0,1]$.
And
$$\begin{aligned}
e^{A^1h}&=\left[\begin{array}{cc}
    1.1052 & 0.1105 \\
    0 & 1.1052
\end{array}\right],
e^{A^2h}=\left[\begin{array}{cc}
    1.1052 & 0 \\
    0.1105 & 1.1052
\end{array}\right],\\
\mathcal{B}^1(h)&=\left[\begin{array}{cc}
    0.2103 & 0.0053\\
    0 & 0.1052
\end{array}\right],
\mathcal{B}^2(h)=\left[\begin{array}{cc}
    0.2103 & 0.1052\\
    0.0107 & 0.2157
\end{array}\right],\\
\mathcal{H}^1(h)&=\left[\begin{array}{cc}
    0.2210 & 0.2157\\
    0.2103 & 0.1052
\end{array}\right],
\mathcal{H}^2(h)=\left[\begin{array}{cc}
    0 & 0.2103\\
    0.1052 & 0.0107
\end{array}\right].
\end{aligned}$$

Then $E^1_1=e^{A^1h},E^2_1=e^{A^2h}$.
$\sigma(\Phi^{1,1}_s)=\sigma(\Phi^{2,2}_s)=\{1.1052\},\xi^1_{1,1}(1)=[0,1],\xi^2_{1,1}(1)=[1,0]$.
Therefore, $\eta^1_{1,1}(1)=v^1_1(1)\otimes\xi^1_{1,1}(1)=[0,1,0,0],\eta^2_{1,1}(1)=v^2_1(1)\otimes\xi^2_{1,1}(1)=[1,0,0,0].$
Since $\eta^2_{1,1}(\Delta^2\otimes\mathcal{B}^2(h))=\mathbf{0}$, the response-layer is uncontrollable itself.\par

However, consider the inter-layer coupling, and one has
$$\mathcal{P}^{2,1}(h)=\left[\begin{array}{cc}
    0.2103 & 0.2103 \\
    0.0107 & 0.1159
\end{array}\right].$$
Solve the equation $\xi(\theta_{1,1}^2I_4-\Phi^{1,1}_s)=\eta^2_{1,1}(1)\Phi^{2,1}_s$, i.e.,
$$\begin{aligned}
&\xi\left[\begin{array}{cccc}
    0 & -0.1105 & 0 & 0 \\
    0 & 0 & 0 & 0 \\
    -0.2210 & -0.2157 & 0 & -0.1105\\
    -0.2103 & -0.1052 & 0 & 0
\end{array}\right]\\=&[1,0,0,0]\left[\begin{array}{cccc}
    0.2103 & 0.2103 & 0 & 0 \\
    0.0107 & 0.1159 & 0 & 0 \\
    0 & 0 & 0 & 0\\
    0 & 0 & 0 & 0
\end{array}\right],
\end{aligned}$$
which leads to $\xi=[-0.9511,\xi_2,0,-1]$, where $\xi_2$ is an arbitrary complex number.
It can be seen that $\eta^1_{1,1}(1)(\Delta^1\otimes\mathcal{B}^1(h))\neq{0}$ and $[\xi(\Delta^1\otimes\mathcal{B}^1(h)),\eta^2_{1,1}(1)(\Delta^2\otimes\mathcal{B}^2(h))]\neq\mathbf{0}$.
According to Theorem \ref{con:two-driver}, the whole networked sampled-data system is controllable.
\end{example}

\subsection{Elimination of Pathological Sampling}\label{sec:exa_2}
\begin{example}\label{exa_1}
Consider a two-layer homogeneous networked sampled-data system with drive-response mode, where both layers consist of two nodes.
Let $\Delta^1=I_2,\Delta^2=O$, and
$$\begin{aligned}
&W^1=\left[\begin{array}{cc}
    0 & 1 \\
    1 & 0
\end{array}\right],
W^2=\left[\begin{array}{cc}
    0 & 1 \\
    4 & 0
\end{array}\right],
D^{2,1}=\left[\begin{array}{cc}
    1 & 0 \\
    0 & 0
\end{array}\right],\\
&A=\left[\begin{array}{cc}
    1 & 1 \\
    -1 & 1
\end{array}\right],
B=\left[\begin{array}{c}
    1\\
    0
\end{array}\right],
C=I_2,
H=I_2.
\end{aligned}$$
Let the sampling period $h=\pi$, one can calculate that
$$\begin{aligned}
&e^{Ah}=\left[\begin{array}{cc}
    -23.1407 & 0 \\
    0 & -23.1407
\end{array}\right],
\mathcal{B}(h)=\left[\begin{array}{c}
    -12.0703\\
    -12.0703
\end{array}\right],\\
&\mathcal{H}(h)=\left[\begin{array}{cc}
    -12.0703 & 12.0703 \\
    -12.0703 & -12.0703
\end{array}\right].
\end{aligned}$$
According to Lemma \ref{lem_homo} and Lemma \ref{eigen_homo}, $\lambda^1_1=1,\lambda^1_2=-1,\lambda^2_1=2,\lambda^2_2=-2$.
$v^1_1(1)=[1,1],v^1_2(1)=[-1,1],v^2_1(1)=[2,1],v^2_2(1)=[-2,1],v^{2,1}_1(1)=[\frac{4}{3},\frac{2}{3}],v^{2,1}_2(1)=[\frac{4}{3},-\frac{2}{3}]$.
And
$$\begin{aligned}
&E^1_1=e^{Ah}+\mathcal{H}(h)=\left[\begin{array}{cc}
    -35.2110 & 12.0703 \\
    -12.0703 & -35.2110
\end{array}\right],\\
&E^1_2=e^{Ah}-\mathcal{H}(h)=\left[\begin{array}{cc}
    -11.0703 & -12.0703 \\
    12.0703 & -11.0703
\end{array}\right],\\
&E^2_1=e^{Ah}+2\mathcal{H}(h)=\left[\begin{array}{cc}
    -47.2814 & 24.1407 \\
    -24.1407 & -47.2814
\end{array}\right],\\
&E^2_2=e^{Ah}-2\mathcal{H}(h)=\left[\begin{array}{cc}
    1.0000 & -24.1407 \\
    24.1407 & 1.0000
\end{array}\right].
\end{aligned}$$
The eigenvalues and corresponding eigenvectors of $E_i^j$, $i=1,2$, $j=1,2$ are
$$\begin{aligned}
&\theta_{1,1}^1=-35.2110+12.0703i, \ \xi_{1,1}^1(1)=[i,1],\\
&\theta_{1,2}^1=-35.2110-12.0703i, \ \xi_{1,2}^1(1)=[-i,1],\\
&\theta_{2,1}^1=-11.0703+12.0703i, \ \xi_{2,1}^1(1)=[-i,1],\\
&\theta_{2,2}^1=-11.0703-12.0703i, \ \xi_{2,2}^1(1)=[i,1],\\
&\theta_{1,1}^2=-47.2814-24.1407i, \ \xi_{1,1}^2(1)=[-i,1],\\
&\theta_{1,2}^2=-47.2814+24.1407i, \ \xi_{1,2}^2(1)=[i,1],\\
&\theta_{2,1}^2=1.0000+24.1407i, \ \xi_{2,1}^2(1)=[-i,1],\\
&\theta_{2,2}^2=1.0000-24.1407i, \ \xi_{2,2}^1(1)=[i,1].
\end{aligned}$$
Then the eigenvectors of the state matrix $\bm{\Phi}_s$ of the whole system can be calculated as follows:
$$\begin{aligned}
\eta^1_{1,1}(1)&=[v^1_1(1),\mathbf{0}]\otimes\xi^1_{1,1}(1)=[i,1,i,1,0,0,0,0],\\
\eta^1_{1,2}(1)&=[v^1_1(1),\mathbf{0}]\otimes\xi^1_{1,2}(1)=[-i,1,-i,1,0,0,0,0],\\
\eta^1_{2,1}(1)&=[v^1_2(1),\mathbf{0}]\otimes\xi^1_{2,1}(1)=[i,-1,-i,1,0,0,0,0],\\
\eta^1_{2,2}(1)&=[v^1_2(1),\mathbf{0}]\otimes\xi^1_{2,2}(1)=[-i,-1,i,1,0,0,0,0],\end{aligned}$$
$$\begin{aligned}
\eta^2_{1,1}(1)&=[v^{2,1}_1(1),v^2_1(1)]\otimes\xi^2_{1,1}(1)\\&=[-\frac{4}{3}i,\frac{4}{3},-\frac{2}{3}i,\frac{2}{3},-2i,2,-i,1],\\
\eta^2_{1,2}(1)&=[v^{2,1}_1(1),v^2_1(1)]\otimes\xi^2_{1,2}(1)\\&=[\frac{4}{3}i,\frac{4}{3},\frac{2}{3}i,\frac{2}{3},2i,2,i,1],\\
\eta^2_{2,1}(1)&=[v^{2,1}_2(1),v^2_2(1)]\otimes\xi^2_{2,1}(1)\\&=[-\frac{4}{3}i,\frac{4}{3},\frac{2}{3}i,-\frac{2}{3},2i,-2,-i,1],\\
\eta^2_{2,2}(1)&=[v^{2,1}_2(1),v^2_2(1)]\otimes\xi^2_{2,2}(1)\\&=[\frac{4}{3}i,\frac{4}{3},-\frac{2}{3}i,-\frac{2}{3},-2i,-2,i,1].
\end{aligned}$$
It is obvious that for every $\eta\in\{\eta^1_{1,1}(1),...,\eta^2_{2,2}(1)\},\eta(\bar{\Delta}\otimes\mathcal{B}(h))\neq\mathbf{0}.$
According to Theorem \ref{con:two-driver_homo_s}, it is shown that the whole networked sampled-data system is controllable.
\end{example}

\subsection{Effect of Multi-rate Sampling on Controllability}\label{sec:exa_34}
\begin{example}\label{exa_3}
Consider the following two-layer networked sampled-data system with drive-response mode. 
The topology of each layer is a chain consisting of two nodes.
Let $w^1_{12}=w^2_{21}=1,\Delta^1=I_2,\Delta^2=O$, sampling period $h=\frac{\pi}{2}$, and
$$\begin{aligned}
&A^1=A^2=\left[\begin{array}{cc}
    1 & 1 \\
    -1 & 1
\end{array}\right],
B^1=B^2=\left[\begin{array}{c}
    1\\
    0
\end{array}\right],
C^1=C^2=I_2,\\
&D^{2,1}=\left[\begin{array}{cc}
    1 & 0 \\
    0 & 0
\end{array}\right],
P^{2,1}=\left[\begin{array}{cc}
    1 & 0 \\
    0 & 0
\end{array}\right],
H^1=H^2=I_2.
\end{aligned}$$
One can calculate that
$$\begin{aligned}
&\mathcal{P}^{2,1}(h)=\left[\begin{array}{cc}
    1.9502 & 0 \\
    -2.9502 & 0
\end{array}\right],
e^{A^Kh}=\left[\begin{array}{cc}
    0 & 4.8105 \\
    -4.8105 & 0
\end{array}\right],\\
&\mathcal{B}^K(h)=\left[\begin{array}{c}
    1.9052\\
    -2.9502
\end{array}\right],
\mathcal{H}^K(h)=\left[\begin{array}{cc}
    1.9502 & 2.9502 \\
    -2.9502 & 1.9502
\end{array}\right],
\end{aligned}$$
$K=1,2.$
Verifying the controllability of system $(\bm{\Phi}_s,\bm{\Psi}_s)$ by the rank of its controllability matrix, one can draw the conclusion that $(\bm{\Phi}_s,\bm{\Psi}_s)$ is controllable for $rank([\bm{\Psi}_s,\bm{\Phi}_s\bm{\Psi}_s,...,\bm{\Phi}_s^7\bm{\Psi}_s])=8$.
However, if there is a limitation on the inter-layer transmission channels so that the inter-layer sampling rate becomes $\pi$, the controllability of the whole system will be damaged, because $(\tilde{\bm{\Phi}}_s,\tilde{\bm{\Psi}}_s)$ is uncontrollable for $rank([\tilde{\bm{\Psi}}_s,\tilde{\bm{\Phi}}_s\tilde{\bm{\Psi}}_s,$ ..., $\tilde{\bm{\Phi}}_s^7\tilde{\bm{\Psi}}_s])=6$.
\end{example}

\begin{example}\label{exa_4}
Consider the following two-layer networked sampled-data system with drive-response mode. 
The topology of each layer is a chain consisted by two nodes.
Let $w^1_{12}=w^2_{21}=1,\Delta^1=I_2,\Delta^2=O$, sampling period $h=\pi$, and
$$\begin{aligned}
&A^1=A^2=\left[\begin{array}{cc}
    1 & 1 \\
    -1 & 1
\end{array}\right],
B^1=B^2=\left[\begin{array}{c}
    1\\
    0
\end{array}\right],
C^1=C^2=I_2,\\
&D^{2,1}=\left[\begin{array}{cc}
    1 & 0 \\
    0 & 0
\end{array}\right],
H^1=H^2=P^{2,1}=I_2.
\end{aligned}$$
One can calculate that
$$\begin{aligned}
&\mathcal{P}^{2,1}(h)=\mathcal{H}^K(h)=\left[\begin{array}{cc}
    -12.0703 & 12.0703 \\
    -12.0703 & -12.0703
\end{array}\right],\\
&e^{A^Kh}=\left[\begin{array}{cc}
    -23.1407 & 0 \\
    0 & -23.1407
\end{array}\right],
\mathcal{B}^K(h)=\left[\begin{array}{c}
    -12.0703\\
    -12.0703
\end{array}\right],
\end{aligned}$$
$K=1,2$.
Verifying the controllability of system $(\bm{\Phi}_s,\bm{\Psi}_s)$ by the rank of its controllability matrix, one can draw the conclusion that $(\bm{\Phi}_s,\bm{\Psi}_s)$ is uncontrollable for $rank([\bm{\Psi}_s,\bm{\Phi}_s\bm{\Psi}_s,...,\bm{\Phi}_s^7\bm{\Psi}_s])=7$.
If the sampling period is decreased to $h'=\frac{\pi}{2}$, then
$$\begin{aligned}
&\mathcal{H}^K(h')=\mathcal{P}^{2,1}(h')=\left[\begin{array}{cc}
    1.9052 & 2.9052 \\
    -2.9052 & 1.9052
\end{array}\right],\\
&e^{A^Kh'}=\left[\begin{array}{cc}
    0 & 4.8105 \\
    -4.8105 & 0
\end{array}\right],
\mathcal{B}^K(h')=\left[\begin{array}{c}
    1.9052\\
    -2.9052
\end{array}\right],
\end{aligned}$$
$K=1,2.$
Re-verify the rank of the controllability matrix, the sampled-data system becomes controllable owing to the faster sampling.
However, if there are limitations of distance or costs on the transmission channels, how to ensure the controllability?
Now let the sampling rate on the control channels be $h'=\frac{\pi}{2}$ and the sampling rate on transmission channels still be $h=\pi$.
The multi-rate sampled-data system $(\tilde{\bm{\Phi}}_s,\tilde{\bm{\Psi}}_s)$ is still controllable for $rank([\tilde{\bm{\Psi}}_s,\tilde{\bm{\Phi}}_s\tilde{\bm{\Psi}}_s,$ ..., $\tilde{\bm{\Phi}}_s^7\tilde{\bm{\Psi}}_s])=8$.
\end{example}

\section{Conclusion}\label{sec:con}
The controllability of multilayer networked sampled-data systems is investigated.
The sampling is periodic on the transmission and control channels, with single-rate and multi-rate patterns both considered.
It is revealed that the multi-rate sampling pattern could have positive or negative effects on the controllability of the whole system.
Necessary or/and sufficient controllability conditions are developed, indicating that the controllability of multilayer networked sampled-data systems is jointly determined by the network topology, external control inputs, inter- and intra-layer couplings, node dynamics and the sampling periods.
Results show that the pathological sampling of single node systems can be eliminated by the network topology and inner couplings.
Owing to the inter-layer structure, the whole system can reach controllability even if the response layer is uncontrollable itself.
In further studies, systems with more complex sampling patterns will be considered, and the non-pathological sampling conditions of networked systems will be further explored.


\balance
\bibliographystyle{IEEEtran}
\bibliography{reference}

\end{document}